\documentclass[12pt,letterpaper]{article}

\usepackage[T1]{fontenc}
\usepackage{lmodern}
\usepackage[utf8]{inputenc}
\usepackage{geometry}
\usepackage{booktabs}
\usepackage{amsmath}
\usepackage{amsthm} 
\usepackage{amsfonts}
\usepackage{dsfont}
\usepackage{cancel}
\usepackage{mathtools}
\usepackage{datetime}
\usepackage{bm}
\usepackage{enumitem}
\usepackage{changepage}
\usepackage{amsthm}
\usepackage{verbatim}
\usepackage{amsmath,amssymb}
\usepackage{graphicx}
\usepackage{emptypage}
\usepackage{mathabx}
\usepackage{newlfont}
\usepackage{tikz-cd}
\usepackage{tikz}
\usepackage[all,cmtip]{xy}
\usepackage[bottom]{footmisc}
\usepackage[titletoc,title]{appendix}
\usepackage{chngcntr}
\usepackage{apptools}
\usepackage{listings}
\usepackage{color}
\usepackage{sgame, tikz} 

\usepackage{kpfonts}  
\usepackage{natbib}
\AtAppendix{\counterwithin{lem}{section}}
\usepackage{caption}
\usepackage{subcaption}
\usepackage{float}
\usepackage{thmtools,thm-restate}
\usepackage{epigraph}
\usepackage{xr-hyper}
\usepackage{hyperref}
\hypersetup{colorlinks,linkcolor=blue,allcolors=blue}
\usepackage{sgamevar}
\hypersetup{urlcolor=black}

\theoremstyle{plain} 
\newtheorem{cor}{Corollary} 
\newtheorem{prop}{Proposition}
 
\newtheorem{theorem}{Theorem}
\newtheorem{lemma}{Lemma}
\theoremstyle{definition} 
\newtheorem{ex}{Example}

\makeatletter
\def\th@remark{%
  \thm@headfont{\bfseries}%
  \normalfont % 
  \thm@preskip\topsep \divide\thm@preskip\tw@
  \thm@postskip\thm@preskip
}
\makeatother
\allowdisplaybreaks
\theoremstyle{remark} 
\newtheorem{rmk}{Remark} 
\newtheorem{alg}{Algorithm} 

\let\emptyset\varnothing

\DeclareMathOperator{\E}{\mathds{E}}
\renewcommand{\P}{\mathds{P}}

\newcommand{\R}{\mathds{R}}

\newcommand{\Y}{\mathcal{Y}}

\DeclareMathOperator{\sign}{sign}

\renewcommand{\epsilon}{\varepsilon}
\renewcommand*\d{\mathop{}\!\mathrm{d}}
\newcommand{\argmax}{\operatornamewithlimits{argmax}}
\newcommand{\argmin}{\operatornamewithlimits{argmin}}

\newcommand{\X}{\mathcal{X}}

\usepackage{setspace}

\theoremstyle{definition}

\reversemarginpar

\newdate{draftdate}{07}{04}{2023}
\usdate

\usepackage[nameinlink]{cleveref}
\crefname{manualasm}{assumption}{assumptions}
\crefname{cor}{corollary}{corollaries}
\crefalias{prop}{proposition}
\crefname{claim}{claim}{claims}
\crefname{ex}{example}{examples}
\crefname{defn}{definition}{definitions}
\crefname{rmk}{remark}{remarks}
\crefname{alg}{algorithm}{algorithms}

\begin{document}

\title{The Simple Economics of Optimal Bundling\thanks{I thank Matthew Gentzkow, Soheil Ghili, Ravi Jagadeesan, Paul Milgrom, Mike Ostrovsky, Ilya Segal, Andy Skrzypacz, Takuo Sugaya, Bob Wilson, and Weijie Zhong for helpful comments and suggestions.}}
\author{Frank Yang\thanks{Graduate School of Business, Stanford University. Email: shuny@stanford.edu.}}
%\date{\today}
\date{\displaydate{draftdate}}
\maketitle
\begin{abstract}
We study optimal bundling when consumers differ in one dimension. We introduce a partial order on the set of bundles defined by (i) set inclusion and (ii) sales volumes (if sold alone and priced optimally). We show that if the undominated bundles with respect to this partial order are nested, then nested bundling (tiered pricing) is optimal. We characterize which nested menu is optimal: Selling a given menu of nested bundles is optimal if a smaller bundle in (out of) the menu sells more (less) than a bigger bundle in the menu. We present three applications of these insights: the first two connect optimal bundling and quality design to price elasticities and cost structures; the last one establishes a necessary and sufficient condition for costly screening to be optimal when a principal can use both price and nonprice screening instruments. 
\\


\end{abstract}
\setcounter{page}{1}
\newpage

\section{Introduction}
A common bundling strategy is to create \textit{nested} bundles: A higher-tier bundle includes all the items in a lower-tier bundle. This strategy, also called tiered pricing, is widely adopted across various industries, including software companies (e.g. Slack), streaming services (e.g. Netflix), and e-commerce platforms (e.g. Shopify).\footnote{For instance, Slack offers four service tiers --- Free, Pro, Business+, and Enterprise Grid --- each of which includes all the features of the previous tier (\citealt{slack2022}).} Several questions arise: When is such a strategy profit-maximizing? Which items to package in one tier versus another tier? How many tiers are optimal? Even though these questions seem to be fundamental, relatively little is known because characterizations for optimal bundling are generally intractable (\citealt{Rochet2003}).\footnote{For instance, the optimal mechanism for selling two goods with additive, independent values remains unknown except for some special cases (\citealt{Carroll2017}).}

In this paper, we answer these questions when consumers differ in one dimension. This dimension could represent, for instance, income levels in retail pricing or enterprise complexity in enterprise pricing. We allow the consumers to have general non-additive values (in particular, heterogeneous preferences over different items, and complementary or substitutable preferences across different items). We allow the seller to have arbitrary costs for producing different bundles. Even though we assume one-dimensional heterogeneity, our bundling problem \textit{is} a multidimensional screening problem because there are multiple instruments (allocations of different bundles). 

We introduce a partial order on the set of bundles: A bundle $b_1$ is \textit{dominated} $(\preceq)$ by another bundle $b_2$ if (i) $b_1$ is a subset of $b_2$ and (ii) $b_1$ has a lower sales volume than $b_2$ when both are sold alone at their respective monopoly prices. This partial order can be readily determined by examining the demand curve for each bundle separately. However, it turns out that this simple partial order, under quasi-concavity assumptions, characterizes the optimal bundling strategy.

Our first main result (\Cref{thm:po}) shows that if the undominated bundles are nested (the \textit{nesting condition}), then nested bundling is optimal. Our second main result (\Cref{thm:char}) characterizes which nested menu is optimal: Selling a given menu of nested bundles is optimal if a smaller bundle in (out of) the menu sells more (less) than a bigger bundle in the menu, where sales volumes are measured when bundles are sold alone. We also present a partial converse  (\Cref{cor:minmax}): If nested bundling is optimal, then the minimal optimal menu must include (i) the best-selling bundle as the smallest bundle in the menu, and (ii) the grand bundle as the least-selling bundle in the menu.\footnote{Both ``best-selling'' and ``least-selling'' refer to sales volumes when bundles are sold alone. Under the optimal menu of bundles, the grand bundle may or may not be bought by the smallest fraction of consumers.} 

We present three applications of our main result. The first two applications connect the key demand parameters (price elasticities) and supply parameters (cost structures) to the structure of optimal bundling and quality discrimination. The third application shows how our bundling result also delivers insights into other multidimensional mechanism design problems, in particular screening with price and nonprice instruments. 

In the first application, we provide a sufficient condition for our nesting condition in terms of price elasticities and further characterize the optimal menu (see \Cref{subsec:demand}). Specifically, we introduce the \textit{union elasticity condition} which states that if the demand curves for two different bundles are both elastic at a certain quantity, then the demand curve for their union is also elastic at that quantity. Assuming zero marginal costs, we show that this condition implies the nesting condition, which then implies the optimality of nested bundling by our main result. We also show that, in a precise sense, the optimal menu can be iteratively constructed by using items with a more elastic demand curve as the ``basic items'' and items with a more inelastic demand curve as the ``upgrade items'' (\Cref{prop:elasticity}). In this case, if sold alone, a large bundle has a sales volume lower than its elastic items but higher than its inelastic items. The full characterization of optimal mechanisms enables comparative statics analysis. We find that as the dispersion of values for one item increases, the monopolist switches the tiers of different items and adopts a menu size that is $U$-shaped in the dispersion parameter (\Cref{prop:rotation}). These comparative statics results differ significantly from those in the standard quality-differentiated goods model, such as in \citet{johnson2006simple}, as our model allows for a much richer set of preferences (see \Cref{rmk:demand}).  

In fact, the quality-differentiated goods model (a la \citealt{mussa1978monopoly}) is a special case of our model where there are no heterogeneous relative preferences (see \Cref{subsec:supply}). Our second application shows that, even in this well-studied special case, our result yields new insights by providing a new characterization of the optimal menu (\Cref{prop:D-hat}). Using this characterization, we can hold the price elasticities of different qualities constant and study the effects of cost structures on the optimal quality design. We show that the optimal menu of qualities can be characterized by the lower increasing envelope of the average cost curve (\Cref{prop:C-check}). This result turns out to (i) generalize a finding by \citet{johnson2003multiproduct} and (ii) refine their intuition about when segmenting markets is profitable (see \Cref{rmk:supply}). On a technical level, the reason why our results are new even for one-dimensional screening problems is that we impose much weaker ``regularity'' assumptions compared to the textbook treatment. This generality allows for various forms of ``bunching'', which are ruled out by standard assumptions but can be characterized by our notion of dominance (see \Cref{rmk:bunching}). 

Our bundling result also provides insights into other multidimensional mechanism design problems. In the third application, building on a connection between bundling and costly screening from \citet{yang2022costly}, we use our main result to characterize when costly screening is optimal for a principal who can use both price and nonprice instruments such as wait time (see \Cref{subsec:screening}). We obtain (see \Cref{prop:screening}) a necessary and sufficient condition for the optimality of costly screening when the agent has ``negatively correlated preferences'' (higher types have higher disutilities), complementing \citet{yang2022costly}, which shows that costly screening is always suboptimal when the agent has ``positively correlated preferences'' (higher types have lower disutilities). Our result shows that, in a precise sense, costly screening is optimal if and only if there exists a costly action such that the elasticity of disutility with respect to consumers' types exceeds a threshold (see \Cref{rmk:costela}).

\paragraph{Discussion of Intuition.}\hspace{-2mm}We now present the key intuition behind our main result (see \Cref{subsec:intuition} for a detailed discussion).

Suppose that we have two items $\{1, 2\}$ and that costs are zero. Suppose further that there are only two possible bundles $\{1\}$ and $\{1, 2\}$ in the market, whose marginal revenue curves (when sold alone) are given in \Cref{fig:MR1}. Because consumers differ in one dimension, we can arrange them along a common quantity axis, as in \Cref{fig:MR1}, with consumers positioned toward the right end having lower values for the bundles. Starting from the right end of the quantity axis, we see that selling to these low-type consumers results in a negative marginal revenue for both bundles $\{1\}$ and $\{1, 2\}$. This is the case until we reach quantity $D^*(\{1,2\})$, the point at which the MR curve of bundle $\{1, 2\}$ crosses zero (i.e. the sales volume when sold alone). From this point on, there is a positive marginal revenue for selling bundle $\{1, 2\}$ to consumers. At the same time, from this point on, the marginal revenue for selling the smaller bundle $\{1\}$ to these consumers is always lower. Noting that the total revenue from any implementable assignments will be the sum of the integrals of the respective MR curves,\footnote{This fact is a consequence of the revenue equivalence theorem (see \Cref{sec:sketch}).} we can see that the optimal strategy is to always sell the bigger bundle $\{1, 2\}$ to consumers located to the left of the quantity $D^*(\{1,2\})$, which can be implemented by simply posting the usual monopoly price for bundle $\{1,2\}$. 

Now, suppose that bundle $\{2\}$ can be also sold in the market. If the MR curve for the smaller bundle $\{2\}$ also relates to the MR curve for the bigger bundle $\{1, 2\}$ in this way, then the same argument implies that selling only the bundle $\{1, 2\}$ is optimal even when all three bundles can be sold in the market. Of course, that need not be the case, as \Cref{fig:MR2} illustrates. But we can still start with the right end of the quantity axis, ``climb up'' the MR curves, and attempt to sell the bundle with the highest marginal revenue at that quantity. This strategy, if implementable, will allow us to attain the \textit{upper envelope} of the MR curves as depicted by the dotted curve in \Cref{fig:MR2}. This strategy may not be implementable, though, as there may not exist prices such that all consumers self-select into exactly the bundles that we want to assign to them. Nevertheless, if the upper envelope of the MR curves happens to be attained by an assignment rule that is \textit{monotone} in the set-inclusion order, then we know that there exist appropriate prices to implement the assignments (similar to the case of selling a single good in which we know that monotone allocations can be implemented by prices). This happens to be the case for \Cref{fig:MR2}, as shown by the depicted bundle assignments. Of course, for the assignment rule to be monotone, the assigned bundles must \textit{necessarily} be nested.

\begin{figure}
    \centering
\hspace{-0.07\linewidth}
\begin{subfigure}[b]{0.48\linewidth} 
\centering
        \includegraphics[scale=0.62]{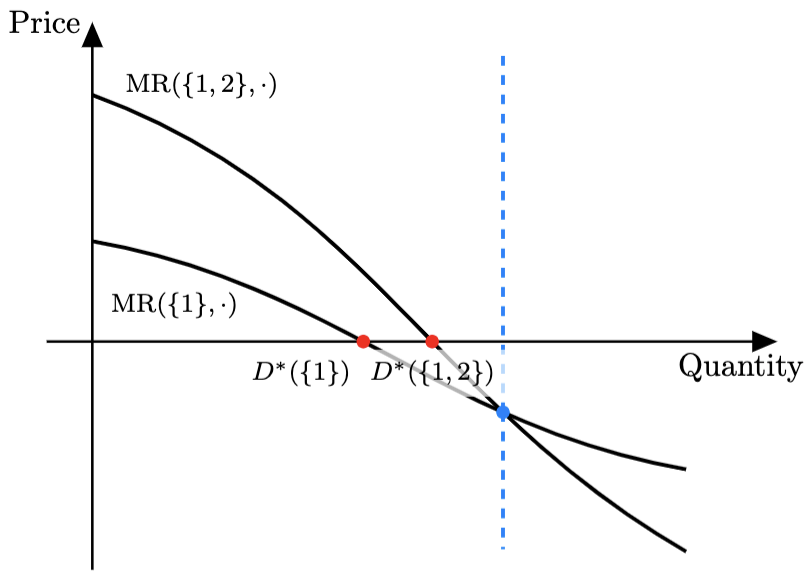}    \caption{MR curves under $\{1\} \preceq \{1, 2\}$\label{fig:MR1}}
\end{subfigure}
\hspace{0.04\linewidth}
\begin{subfigure}[b]{0.48\linewidth}
    \centering
        \includegraphics[scale=0.55]{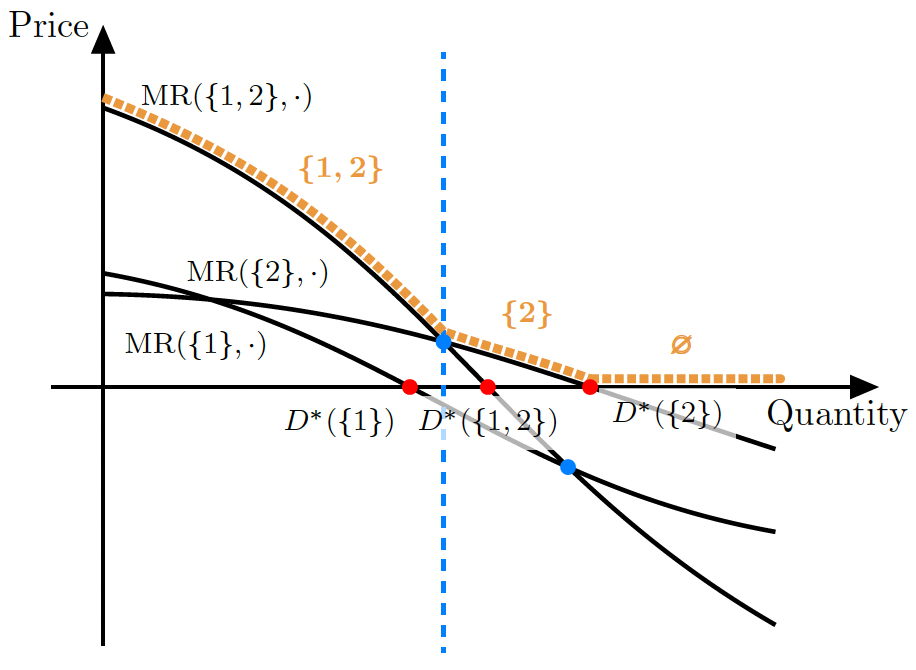}
 \caption{MR curves under $\{1\} \preceq \{1, 2\}$, $\{2\} \not\preceq \{1, 2\}$\label{fig:MR2}}
\end{subfigure}
       \caption{Illustration of the marginal revenue curves}
\end{figure}

A key contribution of this paper is to show that our nesting condition, under quasi-concavity assumptions, is \textit{sufficient} for this strategy to work, which automatically gives us (i) the optimality of nested bundling and (ii) the optimal menu (by ``climbing up'' the MR curves). We briefly explain why our nesting condition is sufficient in the two-item case. Suppose that $D^*(\{1\}) < D^*(\{1, 2\}) < D^*(\{2\})$ (in which case the undominated bundles $\{2\}$ and $\{1,2\}$ are nested). We claim that the \textit{configuration} of MR curves must resemble \Cref{fig:MR2}. Recall that when the revenue functions are quasi-concave, the MR curves cross zero once from above. It turns out that, under the standard quasi-concavity assumption on the revenue function for the incremental bundle, the MR curve of a bigger bundle also crosses the MR curve of a smaller bundle at most once from above.\footnote{This quasi-concavity assumption is stronger than what we actually assume in the model, which is a \textit{local} quasi-concavity condition and allows for multiple crossings of the MR curves (see \Cref{subsec:discussion}).} This implies that, as \Cref{fig:MR1} shows, the two quantities $D^*(\{1\}) < D^*(\{1, 2\})$ must be located in the region where the marginal revenue of selling the incremental bundle is positive (i.e. to the \textit{left} of the vertical dashed line in \Cref{fig:MR1}). At the same time, because the ordering is reversed, as \Cref{fig:MR2} shows, the two quantities $D^*(\{1, 2\}) < D^*(\{2\})$ must be located in the region where the marginal revenue of selling the incremental bundle is negative (i.e. to the \textit{right} of the vertical dashed line in \Cref{fig:MR2}). Thus, the upper envelope of the MR curves must be attained by an implementable allocation rule assigning only nested bundles $\{2\}$ and $\{1, 2\}$. Of course, in the general many-item case, it is impossible to exhaustively list all possible configurations of the MR curves. We provide the key intuition for the general case in \Cref{subsec:intuition}.

\subsection{Related Literature}
There is a substantial literature on multidimensional screening and optimal bundling (beginning with \citealt{stigler1963united}, \citealt{adams1976commodity}, \citealt{McAfee1989MultiproductValues}). A general lesson is that some form of bundling is generically profitable but characterizing optimal bundling strategies turns out to be very difficult (\citealt{armstrong1996multiproduct}, \citealt{Rochet1998}, \citealt{Carroll2017}). Because of this difficulty, relatively little is known about how optimal bundling strategies depend on economic primitives such as price elasticities and cost structures. This paper departs from most of the bundling literature which assumes additive values and multidimensional heterogeneity (\citealt{mcafee1988multidimensional}, \citealt{Manelli2007}, \citealt{pavlov2011optimal}, \citealt{daskalakis2017strong}).\footnote{A few papers also study non-additive values, including \citet{long1984comments}, \citet{armstrong2013more}, \citet{haghpanah2021pure}, \citet{ghili2021characterization}.} By doing so, we are able to connect the empirically relevant economic primitives to the structure of optimal bundling strategies and even perform comparative statics analysis. 

There is a nascent literature on nested bundling (\citealt{bergemann2022optimality}, \citealt{yang2022costly}). \citet{bergemann2022optimality} study the additive case and obtain conditions different from ours. An earlier work by this author introduces a multidimensional screening model with both price and costly nonprice instruments, and studies nested bundling as one application (\citealt{yang2022costly}).\footnote{For the literature on costly screening and money burning, see for example \citet{banerjee1997theory}, \citet{hartline2008optimal}, \citet{condorelli2012money}, \citet{amador2013theory}.} Building on this connection, the present paper proceeds in the reverse direction, by applying the main bundling result to study costly screening. These two papers also differ by imposing non-overlapping sets of assumptions, studying exactly the opposite cases in terms of the preferences over costly actions (see \Cref{subsec:screening}).

There is a long-standing literature on the profitability of price discrimination and more generally the optimality of pure bundling (\citealt{Stokey1979}, \citealt{deneckere1996damaged}, \citealt{anderson2009price}, \citealt{haghpanah2021pure}, \citealt{ghili2021characterization}).\footnote{Note that pure bundling can be thought of as selling only the highest quality version of a product (i.e. no price discrimination).} \citet{ghili2021characterization} introduces the use of single-bundle sales volumes to study the optimality of pure bundling. Our result nests \citet{ghili2021characterization} as a special case (see \Cref{rmk:pb}), which in turn generalizes a long line of inquiry on the profitability of price discrimination (\citealt{salant1989inducing}, \citealt{deneckere1996damaged}, \citealt{anderson2009price}). Importantly, our result pins down the optimal way of price discrimination, beyond the question of whether price discrimination is profitable.

The remainder of the paper proceeds as follows. \Cref{sec:model} presents the model. \Cref{sec:main} presents the main result. \Cref{sec:sketch} sketches the proof of the main result. \Cref{sec:app} presents the applications. \Cref{sec:conclusion} concludes.  \Cref{sec:proof} provides omitted proofs.

\section{Model}\label{sec:model}

A monopolist sells $n$ different goods $\{1, \dots, n\}$ to a unit mass of consumers.

Consumers have types $t \in T:= [\underline{t}, \overline{t}]$. Types are drawn from a distribution $F$ with a continuous, positive density $f$. Type $t$ has value $v(b, t)$ for bundle $b \in \mathcal{B} := 2^{\{1, \dots, n\}}$ with $v(\emptyset, t) = 0$. The value function $v(b, t)$ is (i) nondecreasing in $b$ (in the set-inclusion order), (ii) continuously differentiable in $t$, and (iii) strictly increasing in $t$ whenever $v(b, t) > 0$. For any stochastic assignment $a \in \Delta(\mathcal{B})$, we define $v(a, t) := \E_{b\sim a}[v(b, t)]$. The monopolist incurs cost $C(b)$ to produce bundle $b$ with $C(\emptyset) = 0$. We assume that it is efficient for the highest type $\overline{t}$ to consume all the items: $\argmax_{b} \{v(b, \overline{t}) - C(b)\} = \overline{b}$ where $\overline{b}:=\{1, \dots, n\}$.    

The seller wants to maximize expected profits over all stochastic mechanisms. By the revelation principle, it is without loss of generality to restrict attention to direct mechanisms. Specifically, a \textit{(stochastic, direct) mechanism} is a measurable map $(a, p) : T \rightarrow \Delta(\mathcal{B}) \times \R$ that satisfies the usual incentive compatibility (IC) and individual rationality (IR) conditions:
\begin{align*}
v(a(t), t) - p(t) &\geq v(a(\hat{t}), t) - p(\hat{t}) &&\text{ for all $t, \hat{t}$ in $T$; } \\
v(a(t), t) - p(t) &\geq 0 &&\text{ for all $t$ in $T$}\,.
\end{align*}
Two mechanisms are \textit{equivalent} if they differ on a zero-measure set of types. 

A \textit{menu} is a set of bundles (which we may assume include $\emptyset$).\footnote{To simplify notation, we omit the inclusion of $\emptyset$ in a menu whenever it is clear from the context.} We say that \textit{selling menu} $B$ \textit{is} \textit{optimal} if there exists an optimal mechanism $(a, p)$ such that $a(t) \in B$ for all $t$.\footnote{When an assignment $a(t) \in \Delta(\mathcal{B})$ is deterministic, we also let $a(t)$ denote the assigned bundle.} A menu $B$ is \textit{minimally optimal} if selling $B$ is optimal and selling any $B' \subset B$ is not optimal. A menu $B$ is \textit{nested} if the bundles in $B$ can be ordered by set inclusion. We say that \textit{nested bundling is optimal} if there exists a nested menu $B$ such that selling $B$ is optimal. 

For any bundle $b$, let $F_b$ be the distribution of $v(b, t)$. Let $P(b, q)$ be the \textit{demand curve}:
\[P(b, q) := F^{-1}_b(1 - q)\,.\]
We assume that the \textit{profit function} 
\[\pi(b, q):= (P(b, q) - C(b)) q\]
is strictly quasi-concave in $q \in [0, 1]$.\footnote{For expositional simplicity, whenever we impose strict quasi-concavity of a function $g$ on $[x_1, x_2]$, we assume in addition that $\nabla g(\,\cdot\,) = 0$ at $x \in [x_1, x_2]$ implies $g(x) \geq g(x')$ for all $x' \in [x_1, x_2]$ (i.e. we assume that the FOC is satisfied only at the maximum).} Define $D^*(b)$ as the unique \textit{sales volume} that maximizes the profit function, assumed to lie within $(0, 1)$ for any bundle $b$ with $|b| > 1$.

For any pair of nested bundles $b_1 \subset b_2$, we assume that (i) the \textit{incremental value}
\[v(b_2, t) - v(b_1, t)\]
is strictly increasing in $t$ whenever it is positive, and (ii) the \textit{incremental profit function}
\[\pi(b_2, q) - \pi(b_1, q)\]
is strictly quasi-concave in $q \in [0, \min(D^*(b_1), D^*(b_2))]$ (which is the interval where both profit functions, $\pi(b_1, q)$ and $\pi(b_2, q)$, are increasing).

\subsection{Discussion of Assumptions}\label{subsec:discussion}
\paragraph{Complements and Substitutes.}\hspace{-2mm}The assumptions made here are orthogonal to whether the items are complements or substitutes. To illustrate, consider a simple example where the value for a bundle $b$ is given by $v(b, t) = v(b) \cdot t$. Note that all the above assumptions hold if (i) types $t$ follow a regular distribution in the sense of \citet{Myerson1981} and (ii) $v(b)$ is monotone in the set-inclusion order, regardless of whether the value function $v(b)$ or the monopolist's cost function $C(b)$ exhibit supermodularity or submodularity.

\paragraph{One-dimensional Types.}\hspace{-2mm}At this level of generality, even with one-dimensional types, our model allows for different consumers to have different ordinal rankings over items (see \Cref{ex:two-item}). In fact, every item can be the most preferred item for at least \textit{some} consumers. The main restriction of one-dimensional types in our model is that such horizontal preferences are fixed for a given one-dimensional type $t$. Thus, our model is best suited for capturing settings in which some vertical attribute (such as income) is a good predictor of horizontal preferences for different items.

\begin{figure}[t]
    \centering
\begin{subfigure}[b]{0.48\linewidth}
\centering
        \includegraphics[scale=0.3]{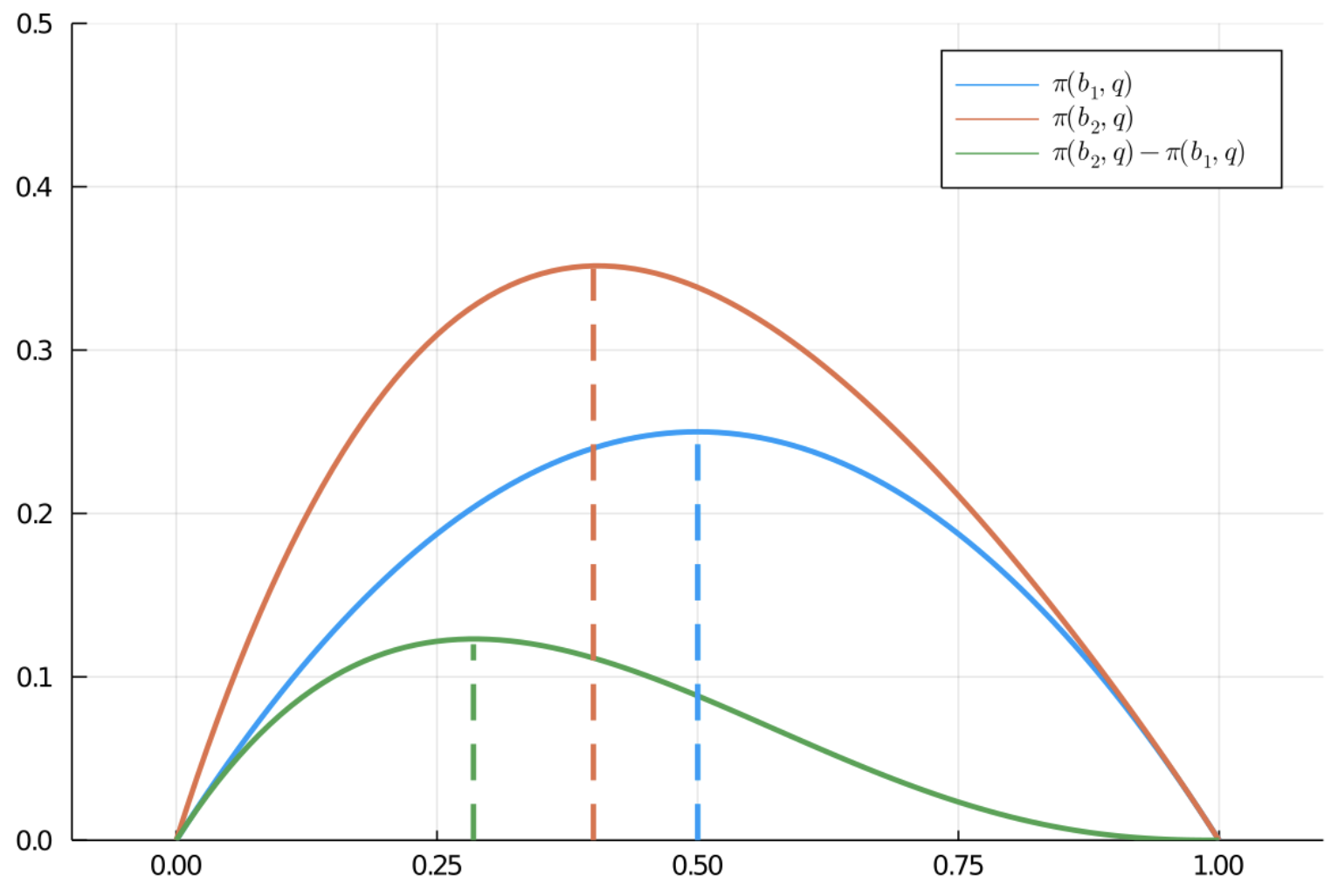}
    \caption{$k = 0$}
\end{subfigure}
\hspace{0.02\linewidth}
\begin{subfigure}[b]{0.48\linewidth}
    \centering
        \includegraphics[scale=0.3]{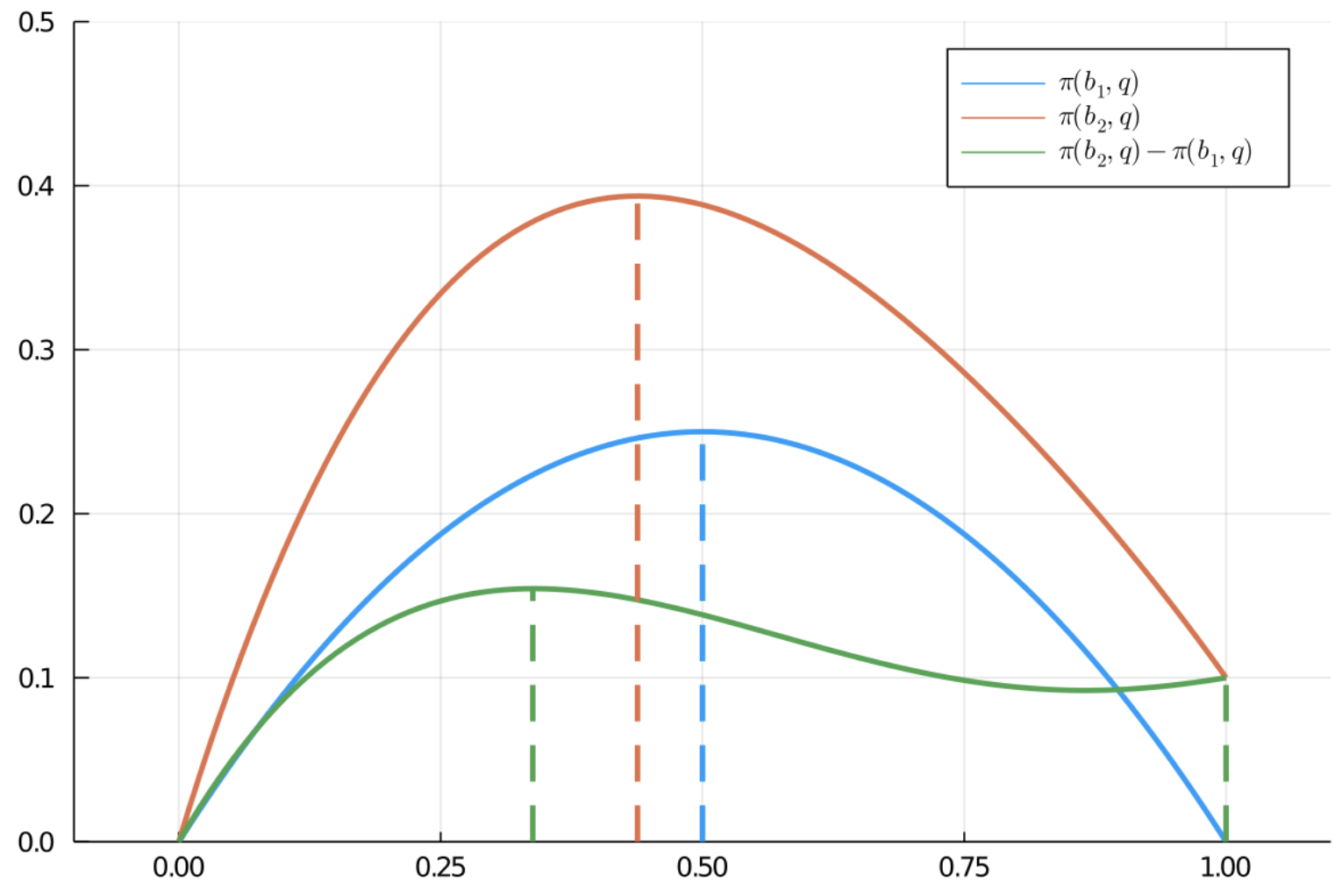}
       \caption{$k = 0.1$}
\end{subfigure}
\caption{Revenue curves given $v(b_1, t) = t$, $v(b_2, t) = t + t^{2.5} + k$, and $t\sim U[0, 1]$}
\label{fig:rev}
\end{figure}

\paragraph{Incremental Profits.}\hspace{-2mm}Consider any pair of nested bundles $b_1 \subset b_2$. Since the incremental value $v(b_2, t) - v(b_1, t)$ is monotone, note that the incremental profit function $\pi(b_2, q) - \pi(b_1, q)$ is equivalent to the profit function of a monopolist optimizing the quantity of the \textit{incremental bundle} $b_2 \backslash b_1$, given the plan of selling every consumer bundle $b_1$.

\paragraph{Local Quasi-concavity.}\hspace{-2mm}We impose only a \textit{local} quasi-concavity condition on the incremental profit function, which states that, within the interval $[0, \min(D^*(b_1), D^*(b_2))]$, the function $\pi(b_2, q) - \pi(b_1, q)$ has at most one peak. This interval is where both profit functions, $\pi(b_1, q)$ and $\pi(b_2, q)$, are increasing (i.e. before both of their peaks). This condition assumes that, within this interval, the sum of an increasing function $\pi(b_2, q)$ and a decreasing function $-\pi(b_1, q)$ is single-peaked. This local quasi-concavity condition is weaker than global quasi-concavity, which in turn is weaker than global concavity. The latter holds when the incremental values follow a regular distribution in the sense of \citet{Myerson1981}. To illustrate, suppose that $v(b_1, t) = t$, $v(b_2, t) = t + t^{2.5} + k$, $C(b_1) = C(b_2) = 0$, and types $t$ follow a uniform distribution $U[0, 1]$. Consider two cases: (i) $k = 0$ and (ii) $k = 0.1$. As shown in \Cref{fig:rev}, the incremental profit function in the first case is globally quasi-concave, whereas the incremental profit function in the second case has two peaks.  However, both cases satisfy our local quasi-concavity assumption.

\section{Main Result}\label{sec:main}

Our main result characterizes (i) when nested bundling is optimal and (ii) which nested menu is optimal. In \Cref{subsec:po}, we introduce a partial order on the set of bundles and show how this partial order characterizes the optimality of nested bundling. In \Cref{subsec:char}, we provide conditions that further characterize the optimality of a given nested menu. In \Cref{subsec:example}, we present a parameterized example to illustrate. In \Cref{subsec:intuition}, we discuss the key intuition behind our results. 

\subsection{Optimality of Nested Bundling}\label{subsec:po}

We define a partial order on the set of bundles $\mathcal{B}$ as follows:
\[b_1 \preceq b_2: \text{ $b_1 \subseteq b_2$ and $D^*(b_1) \leq D^*(b_2)$\,.}\]
A bundle $b$ is \textit{dominated} if there exists $b'\neq b$ such that $b \preceq b'$ and \textit{undominated} otherwise. We say that the \textit{nesting condition} holds if the undominated bundles can be ordered by set inclusion: that is, for any two bundles $b$ and $b'$, 
\[\text{both $b$ and $b'$ are undominated} \implies \text{either $b \subseteq b'$ or $b' \subseteq b$}\,. \tag{\textit{\textbf{Nesting condition}}}\]
\Cref{fig:hasse} illustrates this condition for a three-item example using a diagram, where an upward arrow from $b_1$ to $b_2$ represents $b_1 \preceq b_2$.

\begin{figure}
    \centering
    \includegraphics[scale=0.55]{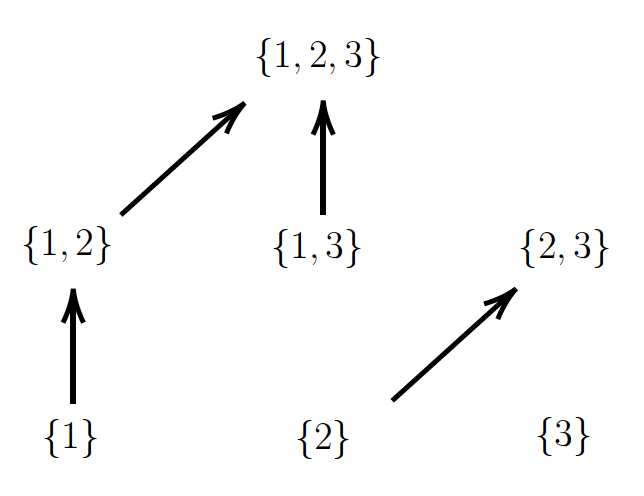}
    \caption{Illustration of the nesting condition for a three-item example}
    \label{fig:hasse}
\end{figure}

\begin{theorem}\label{thm:po}
Under the nesting condition, nested bundling is optimal and every optimal mechanism is equivalent to nested bundling. 
\end{theorem}
\begin{proof}[Proof of \Cref{thm:po}]
See  \Cref{subsec:proof-po}.
\end{proof}

We sketch the proof in \Cref{subsec:sketch-po}. The proof is constructive. It shows that, under the nesting condition, simply selling the set of undominated bundles is optimal. It also provides a simple algorithm to determine the minimal optimal menu and associated prices (see \Cref{alg:alg1} in \Cref{subsec:proof-po}). In the case of two items, the construction shows that the minimal optimal menu is exactly the set of undominated bundles, which we record below:
\begin{cor}\label{cor:two-item}
Suppose that there are two items. Under the nesting condition, the set of undominated bundles is the minimal optimal menu.
\end{cor}

\begin{rmk}\label{rmk:efficiency}
Note that \Cref{thm:po} holds even when the socially efficient allocations require bundles that are not nested. That is, the nesting condition implies the optimality of nested bundling even when the efficient allocations require a far richer set of bundles. Note also that \Cref{thm:po} implies that the optimal mechanism is deterministic. This is in general not true when the nesting condition is not satisfied.
\end{rmk}

\begin{rmk}\label{rmk:eta}
The nesting condition is implied by suitable conditions on price elasticities (when marginal costs are zero). Let $\eta(b, q)$ be the usual price elasticity for bundle $b$ at quantity $q$.\footnote{That is, $\eta(b, q) := \big[\frac{\d \log P(b, q)}{\d \log q}\big]^{-1}$.} In \Cref{subsec:demand}, we show (see \Cref{lem:union}) that a sufficient condition for the nesting condition is the \textit{union elasticity condition} which says that if the demand curves for two bundles are both elastic at a certain quantity $q$, then the demand curve for their union is also elastic at quantity $q$:
\[\eta(b_1, q) < -1 \,\,\text{ and }\,\, \eta(b_2, q) < -1 \implies \eta(b_1 \cup b_2, q) < -1\,.\tag{\textit{\textbf{Union elasticity condition}}}\]
Therefore, as a corollary of \Cref{thm:po}, nested bundling is optimal under the union elasticity condition and zero costs.\footnote{When costs are present, we can simply modify the price elasticity to be $\Tilde{\eta}(b, q):= \big[\frac{\d \log (P(b, q) - C(b))}{\d \log q}\big]^{-1}$.} That is, nested bundling is optimal when bundling results in a demand curve that has a larger elastic region than at least one of the individual demand curves. In \Cref{subsec:demand}, we also characterize (see \Cref{prop:elasticity}) the undominated bundles under the union elasticity condition, and show that the optimal menu can be constructed iteratively by using items with a more elastic demand curve as the ``basic items'' and items with a more inelastic demand curve as the ``upgrade items'', with both measured by the size of their elastic regions.  
\end{rmk}

\begin{rmk}
Both the nesting condition and the union elasticity condition are agnostic to whether the items are complements or substitutes. To illustrate, consider two items and zero costs. Suppose that $v(\{1,2\}, t) = \kappa \cdot \big(v(\{1\}, t) + v(\{2\}, t)\big)$ where $\kappa$ is a positive constant. Depending on the value of $\kappa$, the two items can be complements ($\kappa > 1$), substitutes ($\kappa < 1$), or additive ($\kappa = 1$). However, one can verify that the union elasticity condition always holds in this case, regardless of the value of $\kappa$.
\end{rmk}

\subsection{Characterization of Optimal Menu}\label{subsec:char}

Our second main result provides conditions that further characterize the optimality of a given nested menu.
 
\begin{theorem}\label{thm:char}
Selling a nested menu $B$ is optimal if:
\begin{itemize}
    \item[(i)]  For all $b_1 \in B$, $D^*(b_1) > D^*(b_2)$ for all $b_2 \in B$ such that $b_1 \subset b_2$\,;
    \item[(ii)] For all $b_1 \not\in B$, $D^*(b_1) \leq D^*(b_2)$ for some $b_2 \in B$ such that $b_1 \subset b_2$\,.
\end{itemize} 
Conversely, selling a nested menu $B$ is minimally optimal only if: 
\begin{itemize}
    \item[(i)]  For all $b_1 \in B$, $D^*(b_1) > D^*(b_2)$ for all $b_2 \in B$ such that $b_1 \subset b_2$\,;
    \item[(ii)] For all $b_1 \not\in B$, $D^*(b_1) \leq D^*(b_2)$ for some non-empty $b_2 \in B$\,. 
\end{itemize}
\end{theorem}
\begin{proof}[Proof of \Cref{thm:char}]
See  \Cref{subsec:proof-char}.
\end{proof}
We sketch the proof in \Cref{subsec:sketch-char}. The first statement follows from the proof of \Cref{thm:po} by unpacking the definition of undominated bundles.\footnote{Additionally, we also show that to characterize the optimality of a given nested menu $B$, we can relax our monotone incremental values assumption and locally quasi-concave incremental profits assumption to only apply to pairs of nested bundles $b_1 \subset b_2$ where $b_2 \in B$.} The second statement is proved by constructing a smaller optimal menu when part (i) fails, and a strict improvement when part (ii) fails. The strict improvement when part (ii) fails is \textit{not} a nested-menu mechanism; rather, it perturbs the original mechanism by adding a new option involving probabilistic assignments.  

Let the \textit{best-selling bundle} $b^*$ be the bundle with the highest sales volume when sold alone:
\[b^* := \argmax_{b \neq \emptyset} \{D^*(b)\}\,,\]
which we assume to be unique. A consequence of \Cref{thm:char} is the following result:
\begin{cor}\label{cor:minmax}
Every minimal optimal, nested menu $B$ includes (i) the best-selling bundle $b^*$ as the smallest non-empty bundle in the menu, and (ii) the grand bundle $\overline{b}$ as the least-selling bundle in the menu. 
\end{cor}
In \Cref{cor:minmax}, the term ``least-selling'' also refers to sales volumes when bundles are sold alone. Under the optimal menu of bundles, the grand bundle $\overline{b}$ may or may not be bought by the smallest fraction of consumers. The best-selling bundle $b^*$ and the grand bundle $\overline{b}$ are the two ``extremal'' bundles with respect to our partial order. \Cref{cor:minmax} says that if nested bundling is optimal, then the minimal optimal menu must include these two ``extremal'' bundles and must exclude any bundles dominated by either of the two ``extremal'' bundles.

\Cref{thm:char} can be used to provide sufficient conditions for nested bundling to be strictly suboptimal. For example, a consequence of  \Cref{cor:minmax} is the following result:
\begin{cor}\label{cor:sub}
Suppose that there are two items and that bundle $\{2\}$ is the best-selling bundle. If the optimal profit under menu $\big\{\{2\}, \{1, 2\}\big\}$ is strictly less than the optimal profit under menu $\big\{\{1\}, \{1, 2\}\big\}$, then nested bundling is strictly suboptimal.
\end{cor}

\begin{rmk}\label{rmk:pb}
The second statement of \Cref{thm:char} is a partial converse because (i) it requires the menu $B$ to be minimally optimal and (ii) it does not show $b_1 \subset b_2$ in the second part. However, in the case where $B = \{\overline{b}\}$ (the menu consists only of the grand bundle), neither (i) nor (ii) has bite, and hence \Cref{thm:char} becomes an ``if-and-only-if'' characterization of pure bundling. Thus, our result nests the full characterization of pure bundling by \citet{ghili2021characterization} as a special case.
\end{rmk}

\subsection{An Example}\label{subsec:example}

\begin{ex}\label{ex:two-item}
Suppose that there are two items $\{1, 2\}$ and zero costs. The valuations for each bundle are given by:
\[v(\{1\}, t) = t^{\alpha},\quad v(\{2\}, t) = t^{\beta},\quad  v(\{1, 2\}, t) = t^{\alpha} + t^{\beta} + t^{\gamma}\,.\]
Types $t$ follow a uniform distribution on $[0, 2]$. We fix parameter $\alpha = 1$ and consider two cases: (i) $\gamma = 0.5$ and (ii) $\gamma = 4.5$. In both cases, we vary parameter $\beta$ from $0$ to $2$.\footnote{Note that types $t < 1$ and types $t > 1$ have different ordinal rankings for items $1$ and $2$ whenever $\beta \neq 1$.}

\begin{figure}[t]
    \centering
\begin{subfigure}[b]{0.48\linewidth}
\centering
        \includegraphics[scale=0.3]{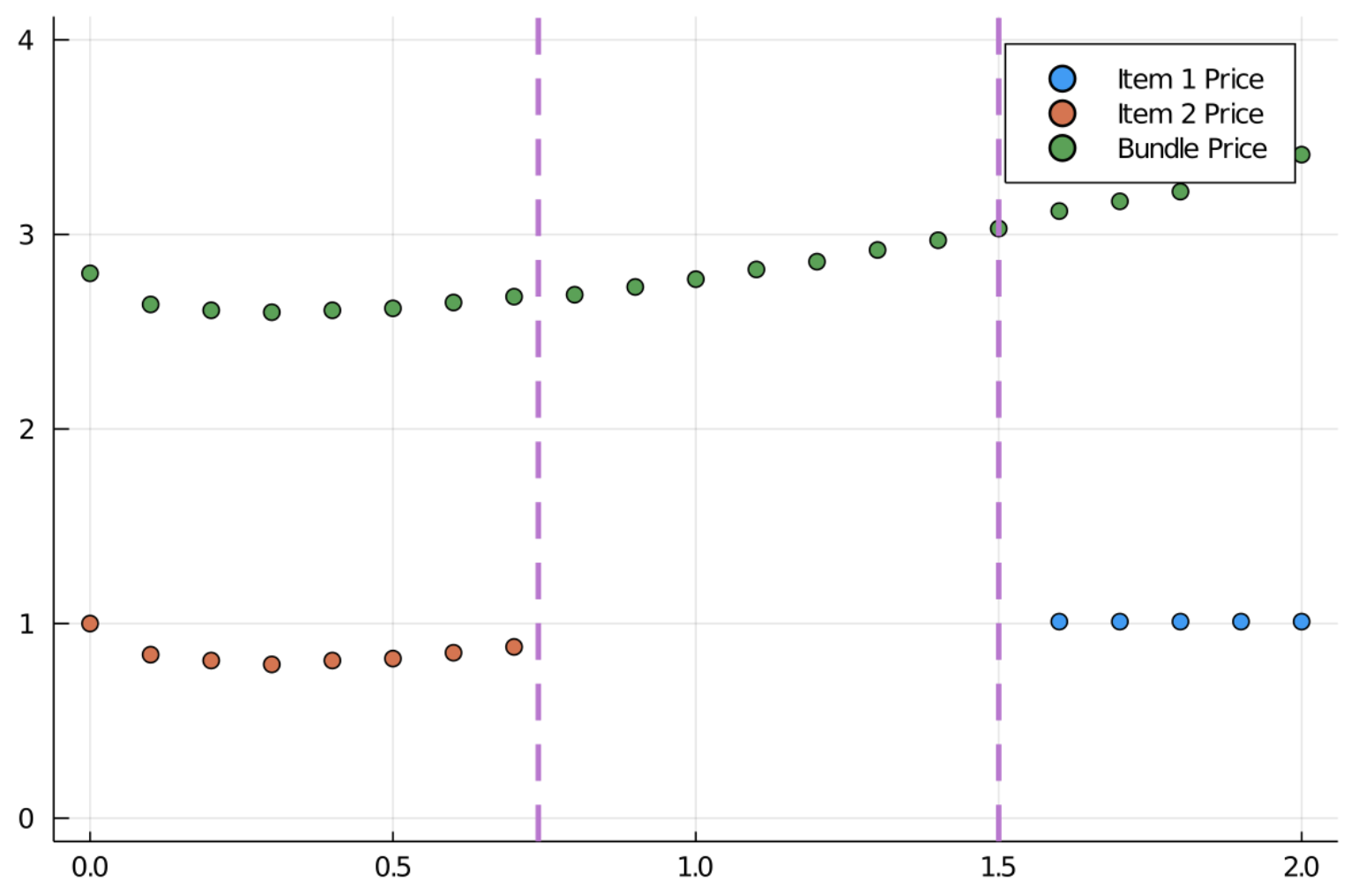}
    \caption{Optimal mechanisms vs. $\beta$ when $\gamma = 0.5$\label{fig:lowgammaP}}
\end{subfigure}
\hspace{0.02\linewidth}
\begin{subfigure}[b]{0.48\linewidth}
    \centering
        \includegraphics[scale=0.3]{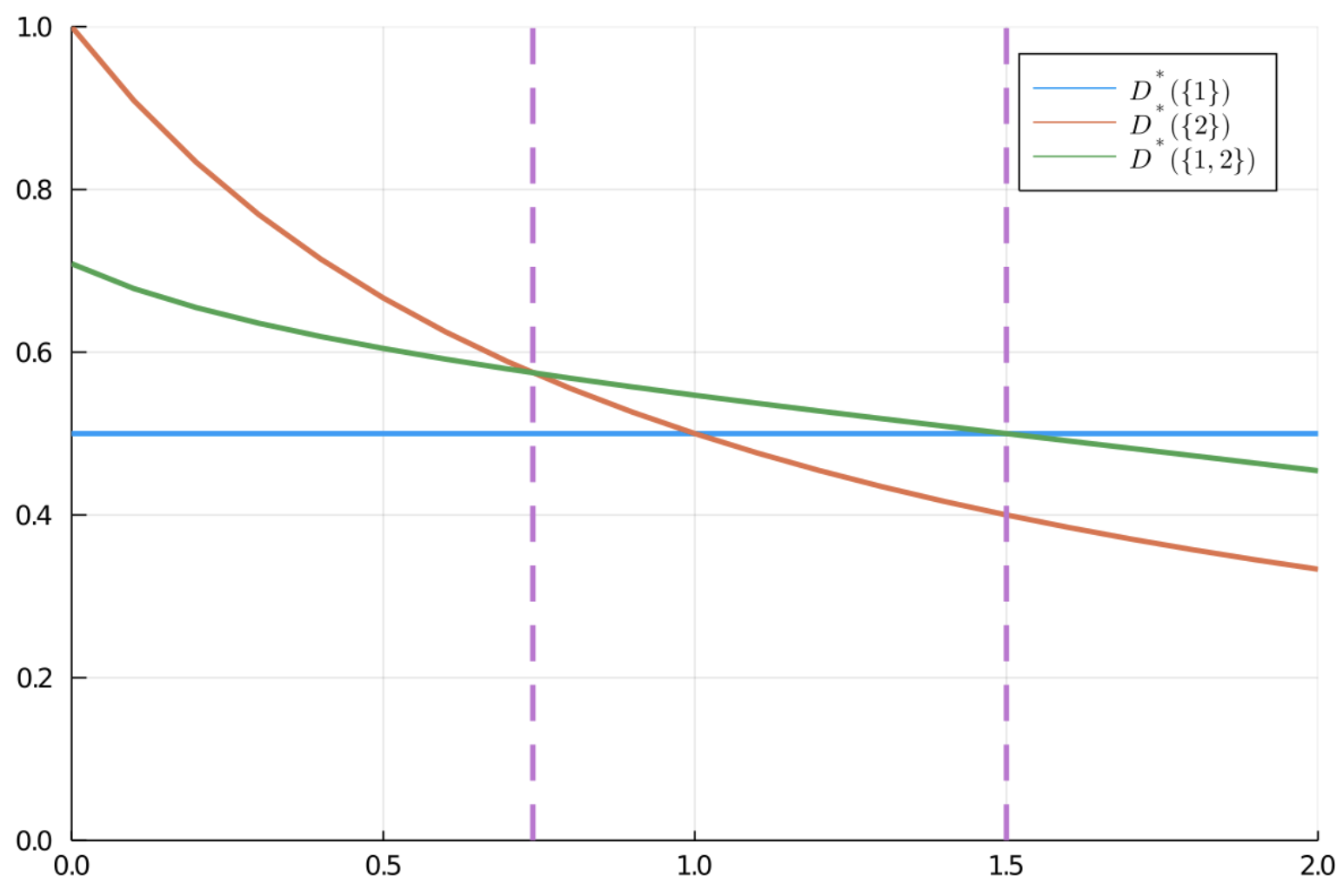}
       \caption{Sales volumes $D^*(b)$ vs. $\beta$ when $\gamma = 0.5$\label{fig:lowgammaD}}
\end{subfigure}
\begin{subfigure}[b]{0.48\linewidth}
\centering
        \includegraphics[scale=0.3]{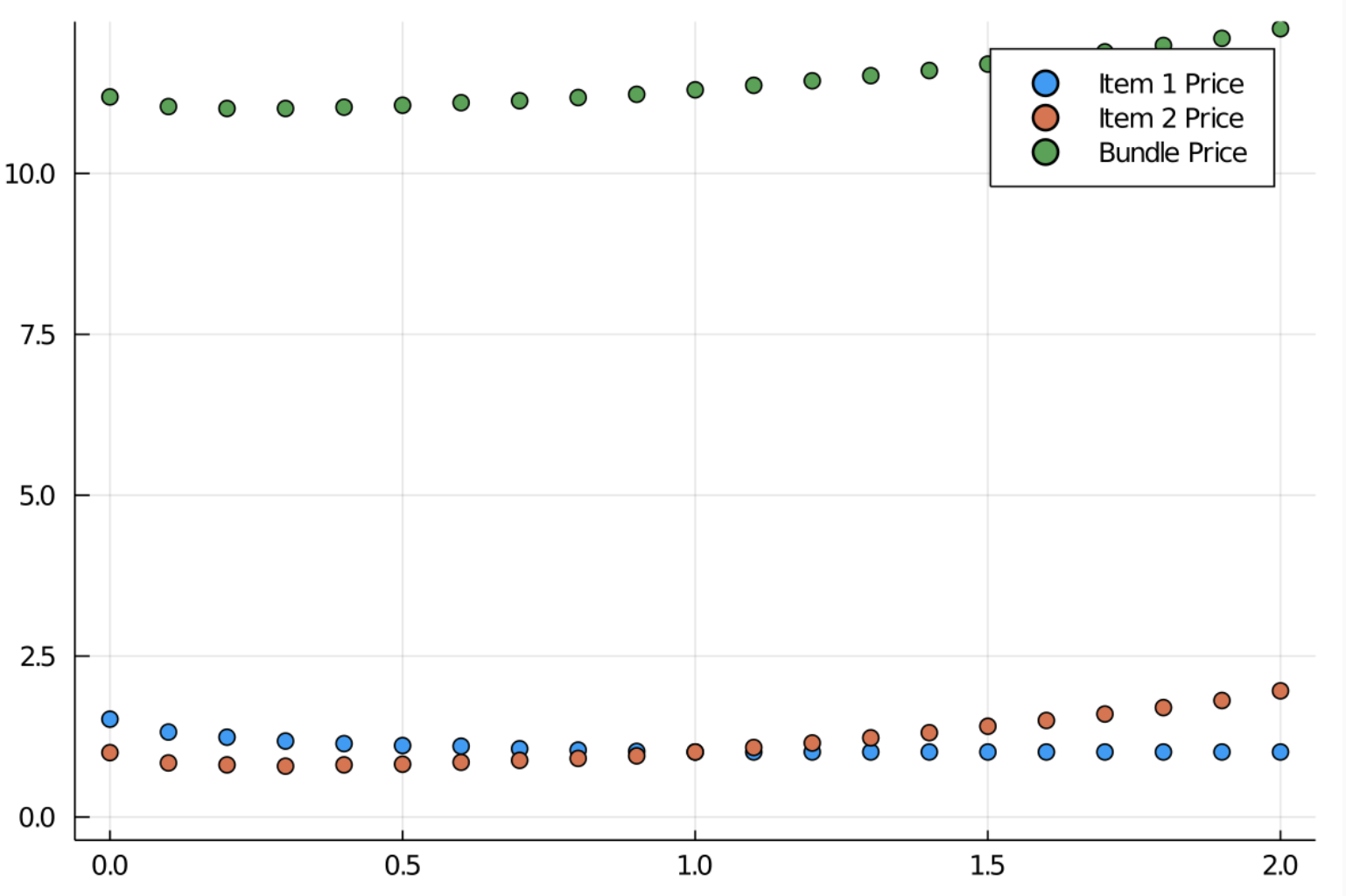}
    \caption{Optimal mechanisms vs. $\beta$ when $\gamma = 4.5$\label{fig:highgammaP}}
\end{subfigure}
\hspace{0.02\linewidth}
\begin{subfigure}[b]{0.48\linewidth}
    \centering
        \includegraphics[scale=0.3]{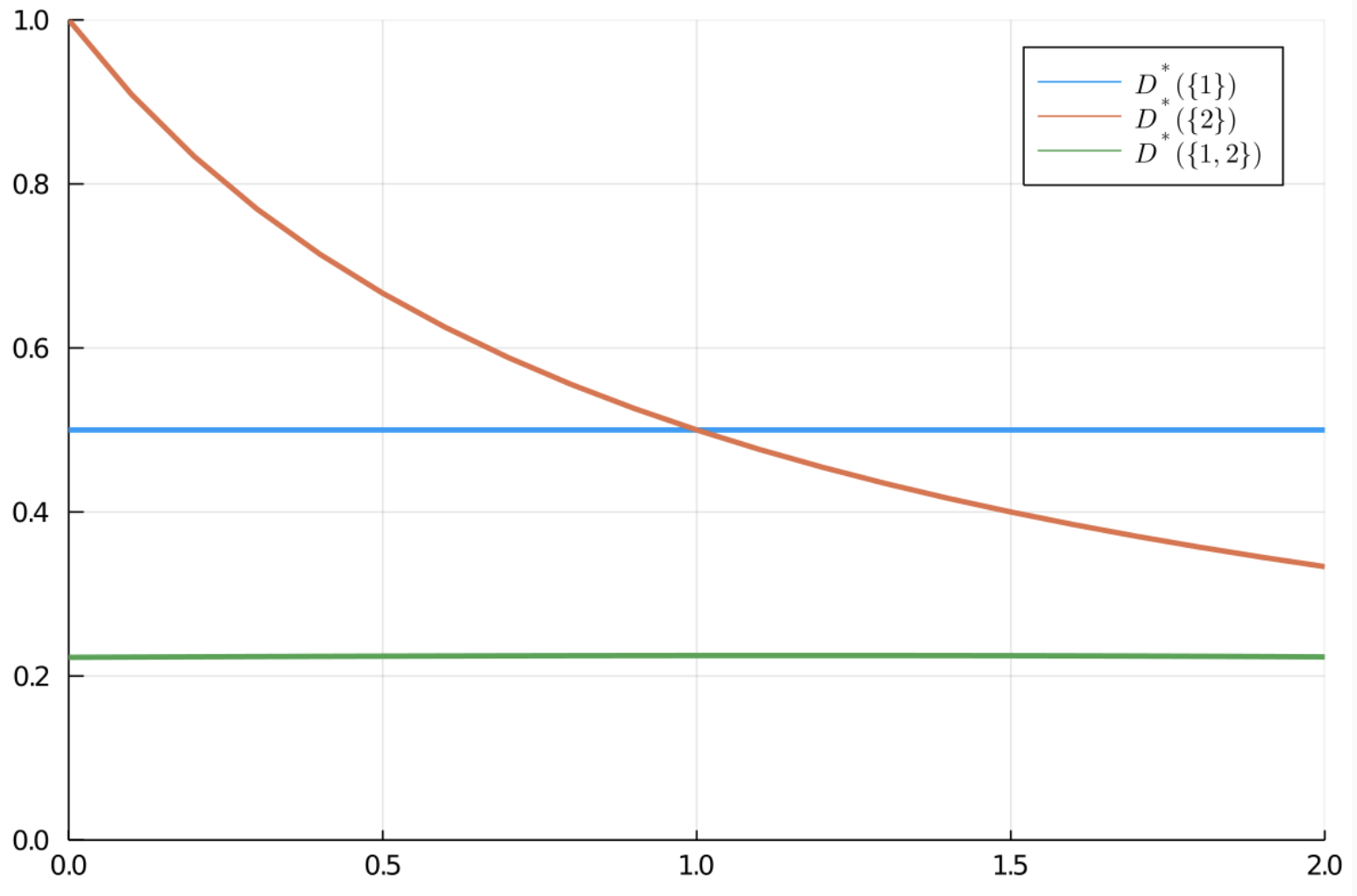}
       \caption{Sales volumes $D^*(b)$ vs. $\beta$ when $\gamma = 4.5$\label{fig:highgammaD}}
\end{subfigure}
\caption{Optimal mechanisms and sales volumes $D^*(b)$ as $\beta$ varies for \Cref{ex:two-item}\label{fig:Ex1}}
\end{figure}

First, consider the case of $\gamma = 0.5$. \Cref{fig:lowgammaP} plots the numerically computed optimal mechanism in terms of prices, as parameter $\beta$ varies in $0.1$ increments. As \Cref{fig:lowgammaP} shows, the optimal mechanism takes different forms as parameter $\beta$ varies. Specifically, the optimal menu is given by: 
\begin{itemize}
\item $\big\{\{2\}, \{1, 2\}\big\}$ when $\beta \in [0, 0.74)$\;;
\item $\big\{\{1, 2\}\big\}$ when $\beta \in [0.74, 1.5]$\;;
\item $\big\{\{1\}, \{1, 2\}\big\}$ when $\beta \in (1.5, 2]$\;.
\end{itemize}
The critical parameter values $\beta = 0.74$ and $\beta = 1.5$
are highlighted by the two vertical dashed lines in \Cref{fig:lowgammaP}. These transitions are characterized by our results. One can verify the union elasticity condition holds for all values of parameter $\beta$.  \Cref{fig:lowgammaD} plots the sales volumes $D^*(b)$ for each of the three bundles as parameter $\beta$ varies. As \Cref{fig:lowgammaD} shows, the undominated bundles are nested for all values of parameter $\beta$. Specifically, the plot can be partitioned into three regions $[0, 0.74)$, $[0.74, 1.5]$, and $(1.5, 2]$. The set of undominated bundles is $\big\{\{2\}, \{1, 2\}\big\}$ in the first region, $\big\{\{1, 2\}\big\}$ in the second region, and $\big\{\{1\}, \{1, 2\}\big\}$ in the third region, coinciding with the optimal menu.

Now, consider the case of $\gamma = 4.5$. In contrast to the previous case, as \Cref{fig:highgammaP} shows, nested bundling is never optimal except at the degenerate parameter value $\beta = 1$ where the two items are indistinguishable. As our results show, this is only possible if the undominated bundles are not nested for any value of parameter $\beta$, which can be seen from \Cref{fig:highgammaD}. 
\end{ex}

\subsection{Discussion of Intuition for \texorpdfstring{\Cref{thm:po}}{} and \texorpdfstring{\Cref{thm:char}}{}
}\label{subsec:intuition}

\subsubsection{Intuition Based on Marginal Revenue Curves}\label{subsubsec:MR}
In this section, we provide the key intuition behind our results, building on the earlier discussion in the introduction. 

Suppose that there are two items $\{1, 2\}$ and zero costs. Consider again the case of $D^*(\{1\}) < D^*(\{1, 2\}) < D^*(\{2\})$. In this case, bundle $\{1\}$ is dominated by bundle $\{1, 2\}$, while bundle $\{2\}$ is not dominated by bundle $\{1, 2\}$. Suppose that the revenue function for any incremental bundle is strictly quasi-concave.\footnote{For the sake of this example we assume that the incremental profit functions are globally quasi-concave. As discussed in \Cref{subsec:discussion}, we need a weaker local quasi-concavity assumption.} This implies that the MR curve of a bigger bundle crosses the MR curve of a smaller bundle at most once from above. \Cref{fig:MR3} illustrates the MR curves under this case as in the introduction.

There are three key observations. First, because of the ordering $D^*(\{1\}) < D^*(\{1, 2\})$, these two quantities must be located in the region where the marginal revenue of upgrading consumers from bundle $\{1\}$ to bundle $\{1, 2\}$ is \textit{positive}. This then implies that if it is profitable to sell a consumer the smaller bundle $\{1\}$, it is even more profitable to \textit{upgrade} the consumer to the bigger bundle $\{1, 2\}$. 

Second, because of the opposite ordering $D^*(\{2\}) > D^*(\{1, 2\})$, these two quantities must be located in the region where the marginal revenue of upgrading consumers from bundle $\{2\}$ to bundle $\{1, 2\}$ is \textit{negative}. This then implies that, after a certain quantity threshold, it is always more profitable to \textit{downgrade} a consumer from the bigger bundle $\{1, 2\}$ to the smaller bundle $\{2\}$. 

Third, with these two operations, we can attain the \textit{upper envelope} of the MR curves by allocating bundles to consumers in a \textit{monotone} fashion such that a higher type consumer receives a bigger bundle in the set-inclusion order. As explained in the introduction, this step is the key to guaranteeing that we can in fact ``climb up'' the MR curves. In the case of two items, this step is self-evident once we recognize that the configuration of MR curves must resemble \Cref{fig:MR3} (as shown in the introduction). However, in the general case, it is impossible to exhaustively list all possible configurations of the MR curves. We will shortly discuss the key intuition behind why our nesting condition is sufficient in the general case. 

Before that, let us consider a case where nested bundling is strictly suboptimal (see \Cref{cor:sub}). By the above arguments, this must be the case where all three bundles are undominated. Without loss of generality, suppose $D^*(\{1, 2\}) < D^*(\{1\}) < D^*(\{2\})$. Now, suppose further that the revenue under menu $\big\{\{2\}, \{1, 2\}\big\}$ is less than the revenue under menu $\big\{\{1\}, \{1, 2\}\big\}$. \Cref{fig:MR4} illustrates the MR curves under this case. Note that, in this case, the upper envelope of the MR curves cannot be attained by a nested menu, so we cannot use the argument of ``climbing up'' the MR curves to find the optimal mechanism. 

There are two opposing forces in this case. On the one hand, it is more profitable to attract the ``medium-type'' consumers than attract the ``low-type'' consumers since the marginal revenue of selling bundle $\{1\}$ to the ``medium-type'' consumers is high enough. On the other hand, it is always possible to attract a small fraction of  the ``low-type'' consumers using bundle $\{2\}$ which can bring in a positive marginal revenue. It turns out that the second force always wins if the monopolist can ration and sell bundle $\{2\}$ with a small probability $\epsilon$. This is because, roughly speaking, the gain from expanding the market this way is on the order of $O(\epsilon)$, whereas the loss from the consumers who no longer purchase bundle $\{1\}$ is on the order of $O(\epsilon^2)$. Intuitively, the reason why the loss is on the higher order is that before introducing bundle $\{2\}$, the monopolist would have already \textit{optimized} the prices for the menu $\big\{\{1\}, \{1, 2\}\big\}$, and hence suffers only a second-order loss for a small perturbation. Thus, nested bundling is strictly suboptimal in this case. 

\begin{figure}
    \centering
\hspace{-0.07\linewidth}
\begin{subfigure}[b]{0.48\linewidth} 
\centering
        \includegraphics[scale=0.62]{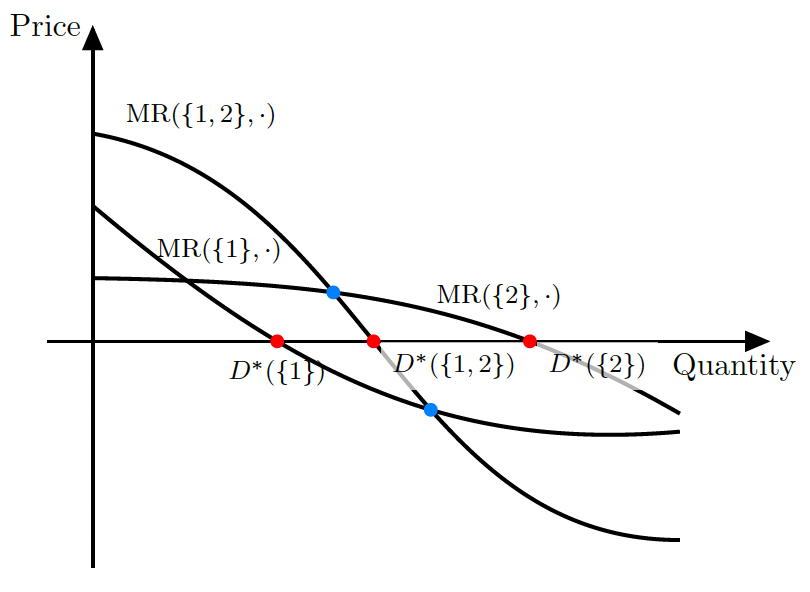}    \caption{MR curves under  $\{1\} \preceq \{1, 2\}$, $\{2\} \not\preceq \{1, 2\}$\label{fig:MR3}}
\end{subfigure}
\hspace{0.04\linewidth}
\begin{subfigure}[b]{0.48\linewidth}
    \centering
        \includegraphics[scale=0.63]{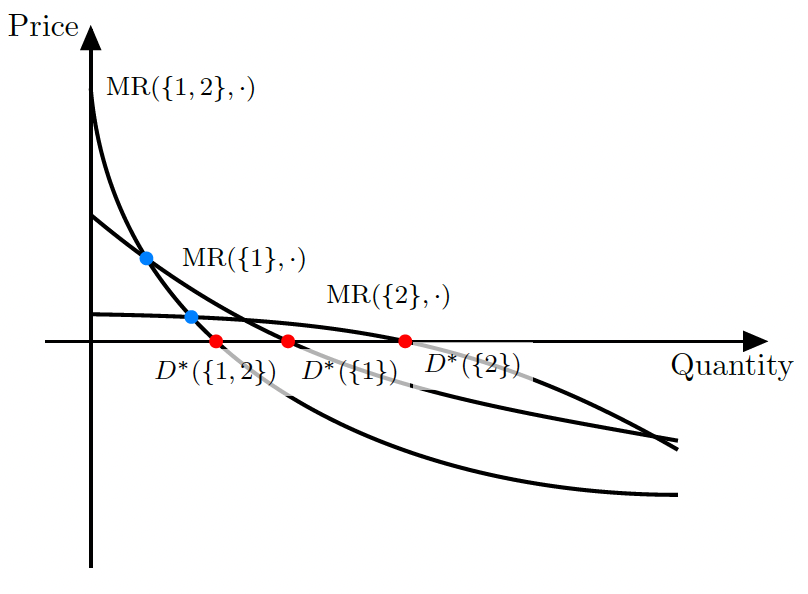}
 \caption{MR curves under $\{1\} \not\preceq \{1, 2\}$, $\{2\} \not\preceq \{1, 2\}$\label{fig:MR4}}
\end{subfigure}
       \caption{Further illustration of the marginal revenue curves}
\end{figure}

\paragraph{Beyond Two-item Cases: Nesting Condition.}\hspace{-2mm}We now explain the key insight that helps understand our results beyond the two-item cases. The intuition as discussed earlier still holds when there are more than two items, but we may run into issues with both the upgrade and downgrade improvements, because these improvements may not be implementable in the price space. To illustrate, suppose that there are three items and that bundle $\{1\}$ is dominated by bundle $\{1, 2\}$. Suppose that we are given an initial allocation rule in the quantity space as depicted in \Cref{fig:s0}. By the previous discussion, we know that if we can upgrade the consumers who are currently consuming bundle $\{1\}$ to bundle $\{1, 2\}$, then we would achieve an improvement. However, this upgrade may not be feasible, because there may not be prices that can support this change in allocations, given that there are higher types who are currently purchasing bundle $\{2, 3\}$, as depicted in \Cref{fig:s1} (highlighted by the double-headed arrow). This is the key difference between our bundling problem and the standard one-dimensional screening problem --- the set of implementable allocation rules is both much richer and much more complex.\footnote{Implementability in multidimensional settings is characterized by \textit{cyclic monotonicity} (see \citealt{rochet1987necessary}) which is much more complex than standard monotonicity conditions.}

\begin{figure}[t]
    \centering
\begin{subfigure}[b]{0.48\linewidth}
\centering
        \includegraphics[scale=0.35]{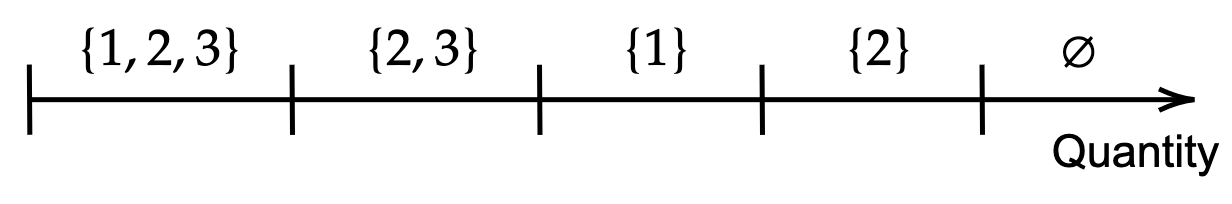}
    \caption{\label{fig:s0}}
\end{subfigure}
\hspace{0.02\linewidth}
\begin{subfigure}[b]{0.48\linewidth}
    \centering
        \includegraphics[scale=0.35]{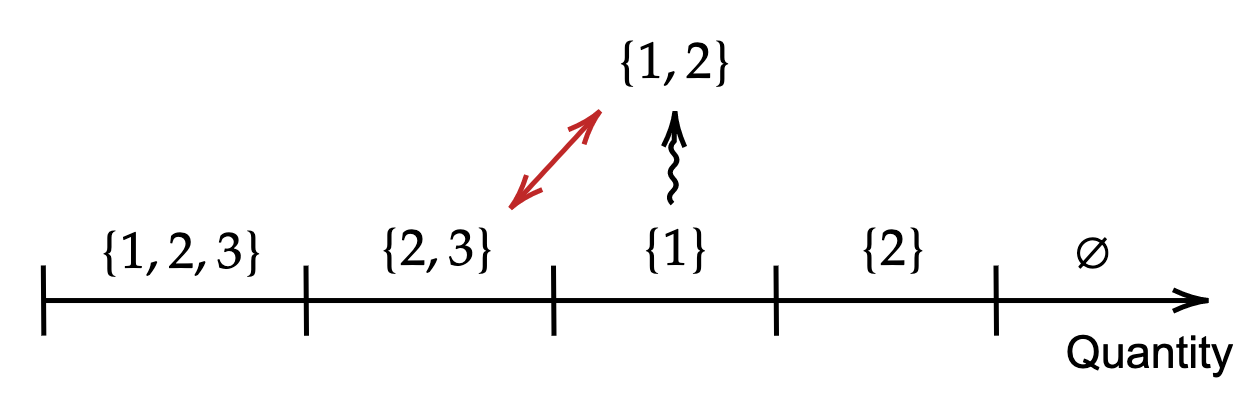}
       \caption{\label{fig:s1}}
\end{subfigure}
\begin{subfigure}[b]{0.48\linewidth}
\centering
        \includegraphics[scale=0.35]{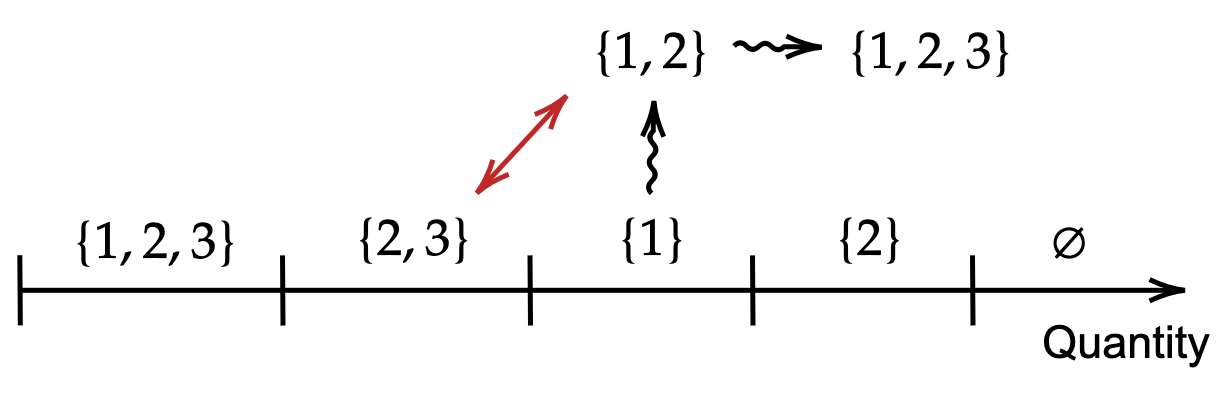}
    \caption{\label{fig:s2}}
\end{subfigure}
\hspace{0.02\linewidth}
\begin{subfigure}[b]{0.48\linewidth}
    \centering
        \includegraphics[scale=0.35]{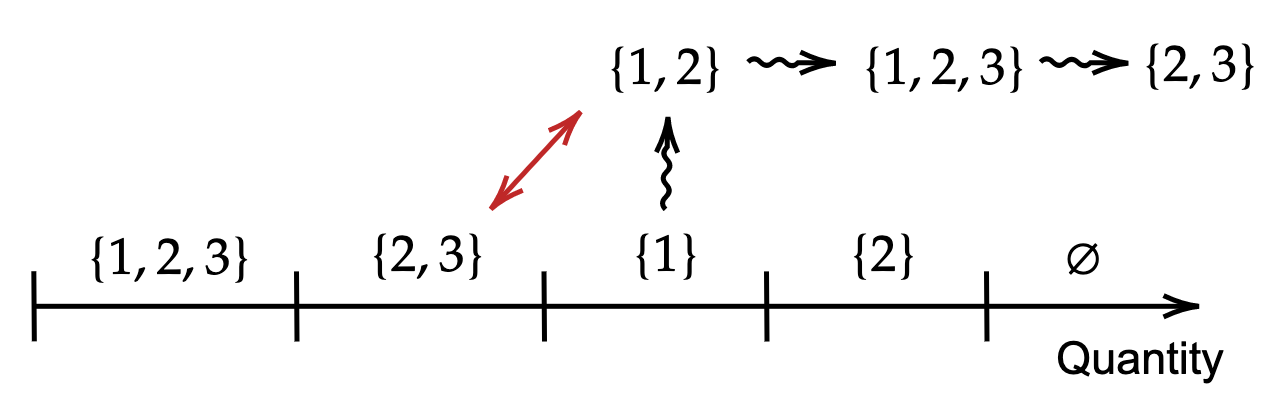}
       \caption{\label{fig:s3}}
\end{subfigure}
\caption{Illustration of the improvement argument for a three-item example}
\end{figure}

The key insight that resolves this problem is the following: Such a potential conflict arises when bundles $b$ and $b'$ are \textit{not} nested, but then the nesting condition implies that at least one of them must be \textit{dominated} (recall that the nesting condition requires the undominated bundles to be nested). In our example, this means that either bundle $\{1, 2\}$ or bundle $\{2, 3\}$ must be dominated. This gives us a way out because we can apply this argument again by going one layer up and further upgrading the consumers from either $b$ or $b'$ to the bundle that dominates one of them. Repeating this process would always result in a pair of nested bundles. 

For our running example, suppose that bundle $\{1, 2\}$ is dominated by bundle $\{1, 2, 3\}$ and bundle $\{2, 3\}$ is undominated, as depicted in \Cref{fig:s2}; so the process in this example terminates in one iteration. Of course, the resulting pair of bundles can be in the ``wrong'' order in the sense that the higher types are assigned the smaller bundle, which we know cannot be implemented by prices. However, when that happens, since the smaller bundle is undominated, we know that if it is ever profitable to downgrade from the bigger bundle to the smaller bundle at some quantity then it is always profitable to downgrade after that quantity. For our running example, suppose that we can further profitably downgrade bundle $\{1, 2, 3\}$ to bundle $\{2, 3\}$, as depicted in \Cref{fig:s3}. The allocation rule is now monotone and hence implementable. Moreover, it is more profitable than the initial allocation rule by construction. Because these arguments can be applied to any initial allocation rule, the upper envelope of the MR curves, under the nesting condition, must be attained by an implementable allocation rule. 

\begin{rmk}
As the above discussion shows, our nesting condition is essential in two ways: (i) it facilities the comparison of marginal revenues, and (ii) it provides a way out from complex implementability constraints by guiding us toward an even more profitable allocation rule that we know is implementable. The actual proof follows these intuitions closely. In addition, the proof considers stochastic mechanisms and shows that, under the nesting condition, randomization cannot increase profit. 
\end{rmk}

\subsubsection{Alternative Intuition Based on Price Elasticities}\label{subsubsec:eta}

In this section, we provide an alternative, price-theoretical intuition based on price elasticities (recall from \Cref{rmk:eta} that the nesting condition is implied by suitable conditions on price elasticities). Suppose that there are two items and zero costs, and that $\eta(\{2\}, q) < \eta(\{1, 2\}, q) <  \eta(\{1\}, q)$ for all quantities $q$. That is, bundle $\{2\}$ has a pointwise more elastic demand curve than bundle $\{1, 2\}$, which in turn has a pointwise more elastic demand curve than bundle $\{1\}$.
This assumption implies that $D^*(\{2\}) > D^*(\{1, 2\}) >  D^*(\{1\})$, but it is much stronger than our union elasticity condition stated in  \Cref{rmk:eta}. 

Let $ \eta(\{1, 2\}, q \mid \{1\})$ be the price elasticity of the demand curve for the \textit{incremental values} for bundle $\{1, 2\}$ given bundle $\{1\}$. We can write 
\[ \eta(\{1, 2\}, q) = \frac{P(\{1, 2\}, q)}{q\cdot\frac{\d}{\d q}P(\{1, 2\}, q)} = \frac{\big(P(\{1,2\}, q) - P(\{1\}, q)\big) + P(\{1\}, q) }{q\cdot\frac{\d}{\d q}\big(P(\{1,2\}, q) - P(\{1\}, q)\big) + q\cdot\frac{\d}{\d q}P(\{1\}, q)}\,.\]
By the mediant inequality,\footnote{That is, for any $\frac{a}{c} < \frac{b}{d}$ such that $c \cdot d > 0$, we have $\frac{a}{c} < \frac{a+b}{c+d} < \frac{b}{d}$.} this implies that
\[  \eta(\{1, 2\}, q \mid \{1\}) < \eta(\{1, 2\}, q) < \eta(\{1\}, q)\,.\]
That is, the demand curve for upgrading from bundle $\{1\}$ to bundle $\{1, 2\}$ is even more \textit{elastic}. In particular, it is profitable to charge a \textit{low} enough price to sell the upgrade option to all consumers in the elastic region of the demand curve for bundle $\{1\}$. Thus, it is profitable to \textit{exclude} bundle $\{1\}$, which is a dominated bundle, as an option from the menu. 

Symmetrically, for bundle $\{2\}$, we have 
\[  \eta(\{1, 2\}, q \mid \{2\}) > \eta(\{1, 2\}, q) > \eta(\{2\}, q)\,.\]
That is, the demand curve for upgrading from bundle $\{2\}$ to bundle $\{1, 2\}$ is even more \textit{inelastic}. In particular, it is profitable to charge a \textit{high} enough price that leaves at least some consumers unserved in terms of the upgrade option. Thus, it is profitable to \textit{include} bundle $\{2\}$, which is an undominated bundle, as an option in the menu. 

\begin{rmk}
Note however that, even in this two-item case, these price-theoretical arguments are incomplete. This is because they do not take into account the presence of the other item in the market when pricing the upgrade from one item to the bundle. In addition, the arguments assume that the elasticities can be pointwise ranked, which is much stronger than our nesting condition. The actual \textit{joint} pricing problem is much more complex and cannot be simply reduced to two separate pricing problems. We emphasize that it is a \textit{consequence} of our results that, under the nesting condition, one can pairwise compare the bundles. As explained in \Cref{subsubsec:MR}, the proof deals with the joint pricing problem by working in the quantity space rather than in the price space. 
\end{rmk}

\section{Proof Sketch for the Main Result}\label{sec:sketch}

In this section, we sketch the proofs for \Cref{thm:po} and \Cref{thm:char}. Following \citet{Myerson1981}, let 
\[\phi(b, t) := v(b, t) - C(b) - \frac{1 - F(t)}{f(t)} v_t(b, t)\]
be the \textit{virtual surplus} function for bundle $b$. As observed by  \citet{bulow1989simple}, this can be interpreted as the \textit{marginal profit} for bundle $b$ evaluated at the quantity such that the marginal consumer is of type $t$. 

A key difference between our problem and the standard one-dimensional mechanism design problem is that we do not have access to a simple characterization of implementable allocation rules. However, the following revenue equivalence theorem continues to hold in our setting: For every implementable allocation rule $a: T \rightarrow \Delta(\mathcal{B})$, the expected profit to the seller is given by 
\[\E\Big [\sum_{b\in\mathcal{B}} a_b(t) \phi(b, t) \Big]\,.\]
The proofs for both \Cref{thm:po} and \Cref{thm:char} involve (i) a series of observations about the collection of functions $\big\{\phi(b, t)\big\}_{b\in \mathcal{B}}$ and (ii) addressing the implementability issues. In the proof of \Cref{thm:po}, we handle implementability by verifying it ex post. In the proof of \Cref{thm:char}, we handle implementability by explicitly constructing an improvement. 

\subsection{Proof Sketch for \texorpdfstring{\Cref{thm:po}}{}}\label{subsec:sketch-po}

The proof sketch involves three steps. In the first step, we establish a key claim about the crossings of $\phi(b_1, t)$ and $\phi(b_2, t)$ for any pair of nested bundles $b_1 \subset b_2$. For any $b_1\subset b_2$, we define the \textit{last crossing type} of the virtual surplus functions for $b_1$ and $b_2$ as 
\[s(b_1, b_2) := \inf_{s}\Big \{s \in T: \phi(b_2, t) > \phi(b_1, t) \text{ for all $t > s$}\Big\}\,.\]
We also define the \textit{last crossing virtual surplus} of $b_1$ and $b_2$ as
\[\chi(b_1, b_2) := \phi(b_1, s(b_1, b_2))\,.\]
The key claim is that for every $b_1 \subset b_2$, we have
\[\sign\Big[\chi(b_1, b_2)\Big]  = \sign\Big[D^*(b_1) - D^*(b_2) \Big]\,.\tag{\textbf{Claim 1}} \label{claim:1}\]

To prove this claim, we first make two basic observations. First, we have 
\[D^*(b_1) \geq D^*(b_2) \iff t^*(b_1) \leq t^*(b_2)\]
where $t^*(b)$ is the cutoff type at which $\phi(b, t)$ single-crosses $0$ from below (which exists by our quasi-concavity assumption on the profit function). Second, we have that $\phi(b_2, t)$ single-crosses $\phi(b_1, t)$ from below on the interval $[\max(t^*(b_1), t^*(b_2)), \overline{t}]$. This is implied by our local quasi-concavity condition, which assumes that the incremental profit function is quasi-concave on the interval $[0, \min(D^*(b_1), D^*(b_2))]$.

Using these two observations, we can prove \ref{claim:1} by making a deeper geometric observation. To illustrate, consider \Cref{fig:crossing}, which depicts the \textit{relative} positions of $\phi(b_1, t)$ and $\phi(b_2, t)$. The key idea is to vary the \textit{horizontal axis} and observe that the sign of $\chi(b_1, b_2)$ is fully determined by the ordering of $t^*(b_1)$ and $t^*(b_2)$, given that $\phi(b_2, t)$ single-crosses $\phi(b_1, t)$ from below on the interval $[\max(t^*(b_1), t^*(b_2)), \overline{t}]$. 

\begin{figure}
    \centering
    \includegraphics[scale=0.54]{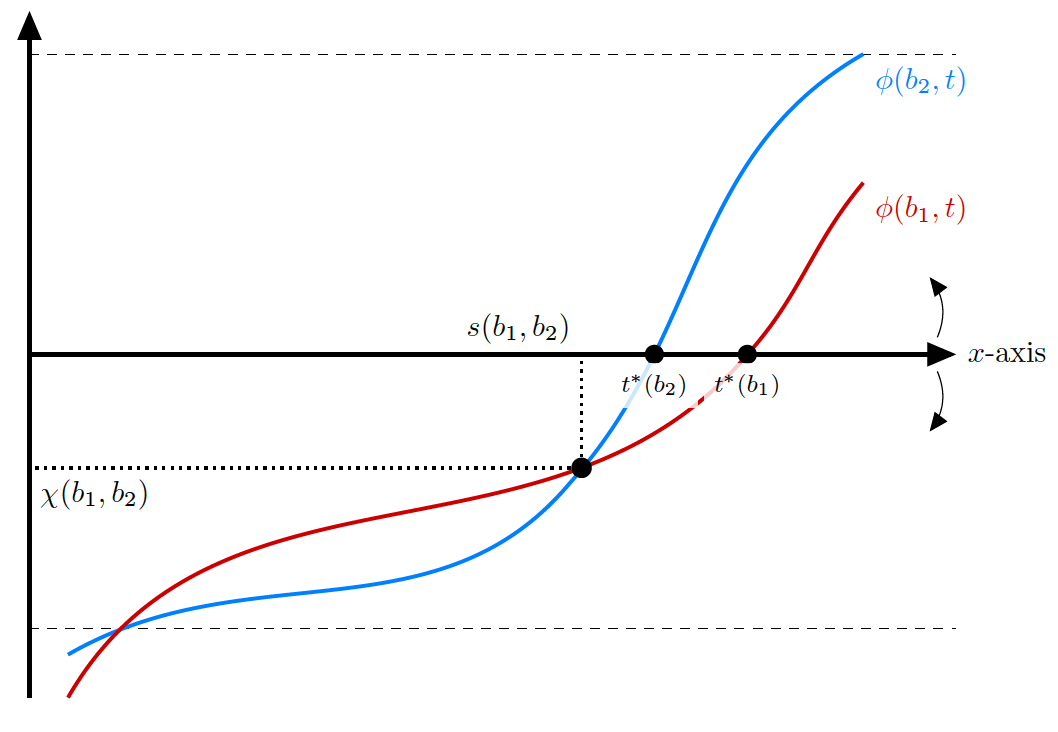}
    \caption{Illustration of possible crossings of $\phi$ when $b_1 \preceq b_2$}
    \label{fig:crossing}
\end{figure}

In the second step, we establish the following claim:  
\[\sup_{a:\, T \rightarrow \Delta(\mathcal{B})} \E\Big [\sum_b a_b(t) \phi(b, t) \Big] = \sup_{a:\, T \rightarrow U} \E\Big [\sum_b a_b(t) \phi(b, t) \Big]\,,\tag{\textbf{Claim 2}} \label{claim:2}\]
where $U$ is the set of undominated bundles (including $\emptyset$). Because the objective of the left-hand side is linear in probabilities, we only need to consider deterministic allocation rules $a: T \rightarrow \mathcal{B}$ for the left-hand side problem.

Now, to prove this claim, we exploit two key facts. First, because $(\mathcal{B}, \preceq)$ is a finite partially ordered set, every dominated bundle must be dominated by an \textit{undominated} bundle. Second, for any $b_1 \subset b_2$, if $\chi(b_1, b_2) \leq 0$, then we must have 
\[\phi(b_1, t) \leq \max\big\{0,\phi(b_2, t) \big\}\,,\]
because by definition $\chi(b_1, b_2) \leq 0$ means that the last time $\phi(b_2, t)$ crosses $\phi(b_1, t)$ from below happens when $\phi(b_1, t)$ is non-positive.

We then combine these two facts as follows. For every dominated bundle $b_1$, we can find an undominated bundle $b_2 \in U$ such that $b_1 \subset b_2$ and $D^*(b_1) \leq D^*(b_2)$, which then implies that $\chi(b_1, b_2) \leq 0$ by \ref{claim:1}. Thus, we have $\phi(b_1, t) \leq \max\big\{0,\phi(b_2, t)\big\}$. Since this argument holds for all dominated bundles, \ref{claim:2} follows.

\begin{figure}
    \centering
    \includegraphics[scale=0.52]{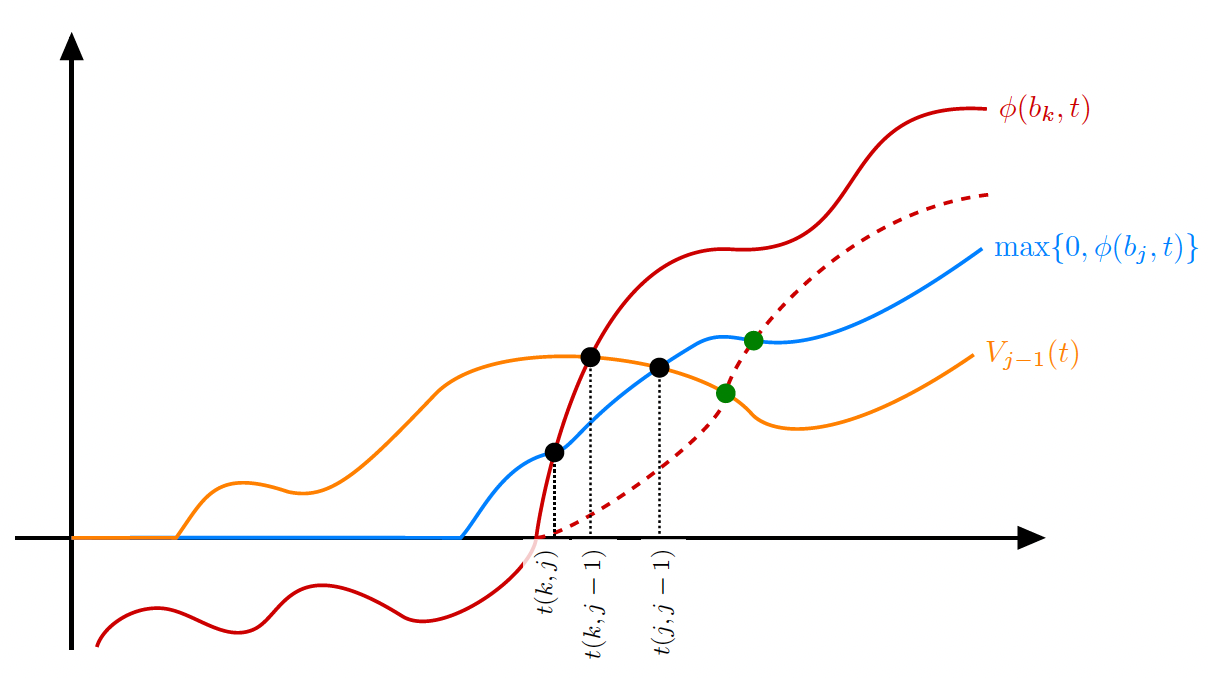}
    \caption{Illustration of possible crossings of $\phi$ in the iterative construction}
    \label{fig:envelope}
\end{figure}

In the third step, we define $a^{\dagger}$ to be the allocation rule induced by the menu $B\subseteq U$ and prices given by \Cref{alg:alg1} in \Cref{subsec:proof-po}. Our goal in this step is to prove the following claim:
\[a^{\dagger} \in  \argmax_{a: \, T \rightarrow U} \E\Big [\sum_b a_b(t) \phi(b, t) \Big]\,.\tag{\textbf{Claim 3}} \label{claim:3}\]
To do so, we use an induction argument. By assumption, the set of undominated bundles $U$ is a nested menu, so we can write $U = \{ \emptyset, b_1, \dots, b_m\}$ such that $b_1 \subset \cdots \subset b_m$. We then construct a sequence of monotone allocation rules $\big\{a^{(j)}\big\}_{j=1,\dots,m}$ using induction such that: (i) the allocation rule $a^{(j)}$ attains the envelope of the virtual surplus functions for the $j$ smallest non-empty bundles in $U$ (and the empty bundle)
\[\phi\big(a^{(j)}(t), t\big) =  \max\Big\{0, \phi(b_1, t), \dots, \phi(b_j, t)\Big\}\,,\]
and (ii) the final allocation rule $a^{(m)}$ coincides with the proposed allocation rule $a^\dagger$. 

The key to this argument is the following geometric observation, which we also prove by induction: For any $k > j$, the function $\phi(b_k, t)$ single-crosses the envelope $V_j(t):= \max\big\{0, \phi(b_1, t), \dots, \phi(b_j, t)\big\}$ from below. To illustrate, consider \Cref{fig:envelope}, which depicts three curves: $\phi(b_k, t)$, $V_{j-1}(t)$, and $\max\{0, \phi(b_j,t)\}$. Note that $V_j(t)$ can be \textit{decomposed} into the envelope of $V_{j-1}(t)$ and $\max\{0, \phi(b_j,t)\}$. By varying $\phi(b_k, t)$ and observing that these three curves single-cross each other, which holds by the inductive hypothesis (on $j-1$), we can establish this single-crossing property. This property allows us to ``paste'' the virtual surplus functions $\phi(b_j, t)$ together one at a time, from the smallest bundle $b_1$ to the largest bundle $b_m$, and trace out the successive envelopes in this process.

Now, recall that the proposed allocation rule $a^\dagger$ is implementable by construction. But then it must be optimal since we have from \ref{claim:2} and \ref{claim:3} that:
\[\E\Big [\sum_b a^\dagger_b(t) \phi(b, t) \Big] =  \sup_{a: \, T \rightarrow U} \E\Big [\sum_b a_b(t) \phi(b, t) \Big] =  \sup_{a:\, T \rightarrow \Delta(\mathcal{B})} \E\Big [\sum_b a_b(t) \phi(b, t) \Big] \,.\]
Thus, under the nesting condition, selling the set of undominated bundles is optimal, and hence nested bundling is optimal. Moreover, as the allocation rule $a^\dagger$ attains the envelope of the virtual surplus functions, the menu $B$ given by \Cref{alg:alg1} must be minimally optimal. The uniqueness part of \Cref{thm:po} follows from a strengthening of \ref{claim:2}.

\subsection{Proof Sketch for \texorpdfstring{\Cref{thm:char}}{}}\label{subsec:sketch-char}

The proof sketch for the first statement is similar to \Cref{subsec:sketch-po} and is hence omitted. The proof sketch for the second statement involves three steps, which we turn to next. 

Suppose that selling a nested menu $B$ (including $\emptyset$) is optimal and selling any proper subset of $B$ is not optimal. Suppose, for contradiction, that either 
\begin{itemize}
    \item[(i)] there exist $b_1 \in B$ and $b_2 \in B$ such that $b_1 \subset b_2$ and $D^*(b_1) \leq D^*(b_2)$, or 
    \item[(ii)] there exists $b_1 \not \in B$ such that $D^*(b_1) > D^*(b_2)$ for all non-empty $b_2 \in B$.
\end{itemize}

In the first step, we define 
\[U_B := \Big\{b \in B: \text{$b$ is not dominated by another bundle $b' \in B$}\Big\}\,.\]
Using similar arguments as in \ref{claim:1} and \ref{claim:2}, we show that there exists an optimal mechanism $(a, p)$ such that $a: T \rightarrow U_B$.

In the second step, we consider Case (i). In this case, the menu $U_B$ must be a strict subset of the menu $B$. Then, since $(a, p)$ is an optimal mechanism, we have that selling menu $U_B$ is also optimal, which contradicts the minimal optimality of menu $B$. Hence, this case cannot occur.

In the third step, we consider Case (ii). We construct a strict improvement in this case (which would then contradict the optimality of menu $B$). Specifically, we propose adding the following option to the menu given by $(a, p)$: A lottery offering bundle $b_1$ with probability $\epsilon$, for a price of $\epsilon \times v(b_1, t^*(b_1))$.

Let us compute the payoff to the monopolist after the addition of this new option. The profit change to the monopolist after adding this new option can be computed as
\[\E\Big[\sum_{b} a'_b(t) \phi(b, t)\Big] - \E\Big[\sum_{b} a_b(t) \phi(b, t)\Big]\]
where $a'$ is the induced allocation rule after the types re-adjust their optimal choices given the new option.

Let $\underline{b}$ be the smallest non-empty bundle in $B$. Note that, by the construction of $(a, p)$, the following holds:
\begin{itemize}
    \item[(i)] all types $t \in [\underline{t}, t^*(b_1))$ will not take this option\,;
    \item[(ii)] all types $t \in [t^*(b_1), t^*(\underline{b}))$ will switch from $\emptyset$ to this option\,;
    \item[(iii)] all types $t \in [t^*(\underline{b}), t^\dagger_\epsilon)$ will switch from $\underline{b}$ to this option, for some threshold $t^\dagger_\epsilon$ that depends on $\epsilon$ (the subscript emphasizes this dependence)\,.
\end{itemize}

\begin{figure}
    \centering
    \includegraphics[scale=0.5]{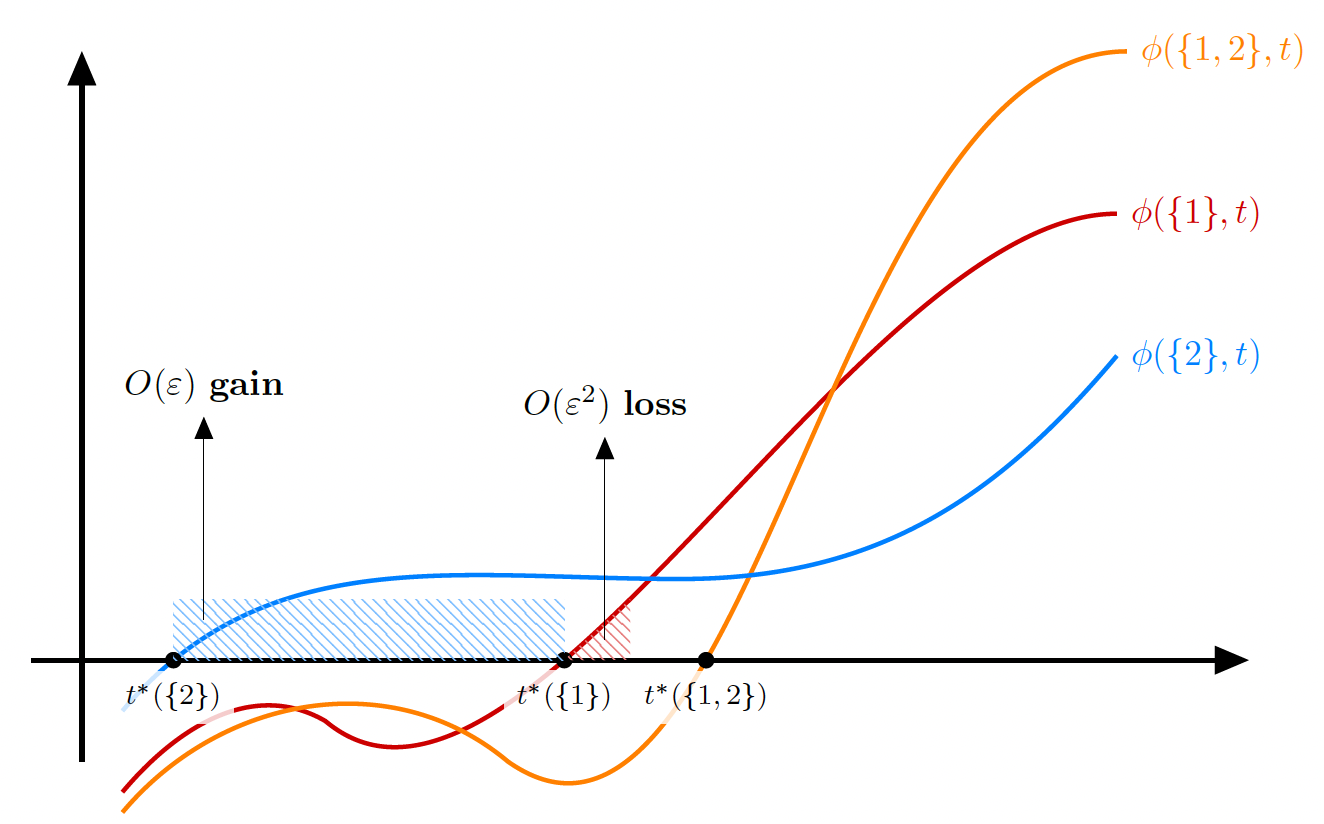}
    \caption{Illustration of the perturbation argument}
    \label{fig:perturb}
\end{figure}

The monopolist makes a gain from the types $t \in [t^*(b_1), t^*(\underline{b}))$ and suffers a loss from the types $t \in [t^*(\underline{b}), t^\dagger_\epsilon)$. It is crucial to compute the gain and the loss in terms of the virtual surplus. Denote them by $\text{Gain}(\epsilon)$ and $\text{Loss}(\epsilon)$. The key claim is the following: For any $\epsilon >0$ small enough, we have 
\[\text{Gain}(\epsilon) > \text{Loss}(\epsilon)\,. \tag{\textbf{Claim 4}}\label{claim:4}\]
Roughly speaking, this is because $\text{Gain}(\epsilon)$ is on the order of $O(\epsilon)$ but $\text{Loss}(\epsilon)$ is on the order of $O(\epsilon^2)$. The argument for \ref{claim:4} relies on a key geometric observation, illustrated in \Cref{fig:perturb} for a two-item example. In particular, observe that the total gain from the types $t \in [t^*(b_1), t^*(\underline{b}))$ forms a \textit{rectangle} whose area varies in $\epsilon$ linearly whereas the total loss from the types $t \in [t^*(\underline{b}), t^\dagger_\epsilon)$ forms a \textit{triangle} whose area varies in $\epsilon$ quadratically.

\section{Applications}\label{sec:app}

In this section, we present three applications of our main result. In the first two applications, we connect empirically relevant demand parameters (price elasticities) and supply parameters (cost structures) to optimal bundling and quality design. In \Cref{subsec:demand}, we study the effect of price elasticities on optimal bundling strategies. In \Cref{subsec:supply}, we study the effect of cost structures on optimal quality design. In the last application, in \Cref{subsec:screening}, we connect costly nonprice screening to optimal bundling. Using our main result, we obtain a necessary and sufficient condition for costly screening to be optimal. 

\subsection{Optimal Bundling Strategies and Price Elasticities}\label{subsec:demand}

We first introduce a sufficient condition for the nesting condition in \Cref{thm:po}, in terms of price elasticities. Let 
\[\eta(b, q) := \Big[\frac{\d \log P(b, q)}{\d \log q}\Big]^{-1}\]
be the usual \textit{price elasticity} for bundle $b$ evaluated at quantity $q$. We say that the \textit{union elasticity condition} holds if for any bundles $b_1$ and $b_2$, we have 
\[\eta(b_1, q) < -1 \,\,\text{ and }\,\, \eta(b_2, q) < -1 \implies \eta(b_1 \cup b_2, q) < -1\,.\tag{\textit{\textbf{Union elasticity condition}}}\]
That is, if the demand curves for two bundles are both elastic at a certain quantity $q$, then the demand curve for their union is also elastic at quantity $q$.

\begin{lemma}\label{lem:union}
Under zero costs, the union elasticity condition implies the nesting condition.   
\end{lemma}
\begin{proof}[Proof of \Cref{lem:union}]
Suppose for contradiction that there exist undominated bundles $b_1$ and $b_2$ that are not nested. Then $b_1 \subset b_1 \cup b_2$ and $b_2 \subset b_1 \cup b_2$. Since costs are zero, the price elasticity for any bundle $b$ at quantity $D^*(b)$ must satisfy 
\[\eta(b, D^*(b)) = -1\,.\]
Note that the elasticity curve $\eta(b, \,\cdot\,)$ single-crosses $-1$ from below because the MR curve $\text{MR}(b, \,\cdot\,)$ single-crosses $0$ from above.\footnote{To see this, recall that $\text{MR}(b, q) = P(b, q) \big[1 + \frac{1}{\eta (b, q)}\big]$.} Then, since $b_1$ and $b_2$ are undominated, we have
\[\eta\big(b_1, D^*(b_1 \cup b_2)\big) < -1 \,\,\,\text{ and }\,\,\, \eta\big(b_2, D^*(b_1 \cup b_2)\big) < -1\,,\]
which, by the union elasticity condition, implies that 
\[\eta\big(b_1 \cup b_2, D^*(b_1 \cup b_2)\big) < -1\,.\]
A contradiction.  
\end{proof}

\begin{rmk}\label{rmk:neteta}
When costs are present, we can simply modify the price elasticity $\eta(b, q)$ to be $\Tilde{\eta}(b, q):= \big[\frac{\d \log (P(b, q) - C(b))}{\d \log q}\big]^{-1}$ to incorporate arbitrary costs.
\end{rmk}

Now, applying our main result, we can fully characterize the optimal menu under the union elasticity condition. To do so, we order the bundles according to their sales volumes $D^*(b)$ when sold alone and define $b^*_i$ as the \textit{$i$-th best-selling bundle}, with ties broken arbitrarily. Then, we have the following result: 

\begin{prop} \label{prop:elasticity}
Under the union elasticity condition and zero costs, selling the following nested menu is optimal: 
\[\Big\{b_1^*, \,\,\,  b_1^* \cup b_2^*, \,\,\,  b_1^* \cup b_2^* \cup b_3^*, \,\,\, \dots, \,\,\, \overline{b}\Big\}\,.\]
\end{prop}
\begin{proof}[Proof of \Cref{prop:elasticity}]
By \Cref{lem:union} and \Cref{thm:po}, nested bundling is optimal. Let $B$ be the proposed menu. By \Cref{thm:char}, it suffices to show that any (non-empty) bundle $b \not \in B$ is dominated. We start by showing that for all $i$, we have
\[D^*(b^*_1 \cup \dots \cup b^*_i) \geq D^*(b^*_i)\,.\]
We prove this by induction on $i$. The base case $i = 1$ is trivial. For the inductive step, suppose that the claim holds for $i-1$. Now, observe that  
\[D^*(b^*_1 \cup \dots \cup b^*_i) \geq \min\Big\{D^*(b^*_1 \cup \dots \cup b^*_{i-1}),\, D^*(b^*_i)\Big\} \geq \min\Big\{D^*(b^*_{i-1}),\, D^*(b^*_i)\Big\} = D^*(b^*_i)\,,\]
where (i) the first inequality follows from the union elasticity condition and the argument used in the proof of \Cref{lem:union}, (ii) the second inequality follows from the inductive hypothesis, and (iii) the last equality follows from the definition of $b^*_i$ and $b^*_{i-1}$. This proves the inductive step. 

Now, fix any $b \notin B$. There exists some index $j$ such that $b = b^*_j$. Since $b \notin B$, we have 
\[b = b^*_j \subset b^*_1 \cup \dots \cup b^*_j\,.\]
But we also have 
\[D^*(b) = D^*(b^*_j) \leq  D^*(b^*_1 \cup \dots \cup b^*_j)\,.\]
Thus, bundle $b$ is dominated, completing the proof.
\end{proof}

\begin{rmk}
Under the union elasticity condition, \Cref{prop:elasticity} provides a simple recipe for constructing the optimal menu: (i) arrange all bundles in descending order based on their sales volumes, and (ii) successively merge them, excluding duplicates. The optimal mechanism iteratively creates nests such that items with a more elastic demand curve become the ``basic items'' and items with a more inelastic demand curve become the ``upgrade items'', with both measured by the size of their \textit{elastic regions}. Note also that this mechanism sorts the \textit{bundles}, rather than the \textit{items}, by their sales volumes. This construction fully accounts for the complementarity or substitutability patterns across different items. 
\end{rmk}

\paragraph{Comparative Statics.}\hspace{-2mm}Price elasticities can be affected by advertising and marketing, which can act as a demand rotation in the sense of \citet{johnson2006simple}. Using \Cref{prop:elasticity}, we can analyze the comparative statics of optimal bundling given a sequence of demand rotations. Suppose that there are two items and zero costs. Consider a family of demand systems indexed by parameter $s \in \R$, with $\eta(b, q; s)$ denoting the price elasticities and $D^*(b; s)$ denoting the sales volumes. We use the following notion of demand rotations: There is a sequence of \textit{(clockwise, sales-ordered) demand rotations} for item $i$ if for all $s < s'$
\[D^*(\{i\}; s') \leq D^*(\{i\}; s), \qquad  D^*(\{j\}; s') = D^*(\{j\}; s) , \qquad D^*(\{1,2\}; s') \leq D^*(\{1,2\}; s),\]
and 
\[  D^*(\{i\}; s) \leq  D^*(\{1,2\}; s)  \implies   D^*(\{i\}; s') \leq D^*(\{1,2\}; s') \,. \]
That is, as parameter $s$ increases, the demand curve for item $i$ and the demand curve for bundle $\{1, 2\}$ become more inelastic in the sense of a smaller elastic region.\footnote{A sufficient (but far from necessary) condition is that, as parameter $s$ increases, the demand curves become pointwise more inelastic.} The last condition ensures that the indirect change in the demand curve for bundle $\{1, 2\}$ is smaller than the direct change in the demand curve for item $i$. To state our result, we define the \textit{tier of item} $i$ in a nested menu $B:= \{b_1, \dots, b_m\}$, where $b_1 \subset \cdots \subset b_m$, as the index of the smallest bundle in $B$ that includes item $i$, denoted by $r_i(B)$.

\begin{prop}\label{prop:rotation}
Suppose that there are two items and zero costs and that the union elasticity condition holds for all $s$. Let $B^*(s)$ be the minimal optimal menu. Then, in a sequence of demand rotations for item $i$, we have that: 
\begin{itemize}
    \item[(i)] the tier of item $i$ in the optimal menu $r_i(B^*(s))$ is nondecreasing in $s$\,;
    \item[(ii)] the tier of item $j \neq i$ in the optimal menu $r_j(B^*(s))$ is nonincreasing in $s$\,;
    \item[(iii)] the size of the optimal menu $|B^*(s)|$ is quasi-convex in $s$\,.
\end{itemize}
\end{prop}

\begin{proof}[Proof of \Cref{prop:rotation}]
Without loss of generality, suppose that there is a sequence of demand rotations for item $2$. By \Cref{prop:elasticity}, nested bundling is always optimal at any parameter $s$. By \Cref{cor:two-item}, the minimal optimal menu $B^*(s)$ equals the set of undominated bundles. To prove claim (i), observe that it suffices to show that if $B^*(s) = \big\{\{1\}, \{1, 2\}\big\}$, then $B^*(s')$ must be $\big\{\{1\}, \{1, 2\}\big\}$ for any $s < s'$. Suppose not. Then, for some $s < s'$, we have 
\[D^*(\{1, 2\}; s') \geq D^*(\{1\}; s') = D^*(\{1\}; s) > D^*(\{1, 2\}; s)\,\,,\]
which is impossible by our notion of demand rotations. 

To prove claim (ii), observe that it suffices to show that if $B^*(s') = \big\{\{2\}, \{1, 2\}\big\}$, then $B^*(s)$ must be $\big\{\{2\}, \{1, 2\}\big\}$ for any $s < s'$. Suppose not. Then, for some $s < s'$, we have 
\[D^*(\{2\}; s) \leq D^*(\{1, 2\}; s)  \,, \qquad D^*(\{2\}; s') > D^*(\{1, 2\}; s')\,,\]
which is impossible by our notion of demand rotations. 

To prove claim (iii), observe that it suffices to show that it cannot be $|B^*(s)| = 1, |B^*(s')| = 2, |B^*(s'')| = 1$ for any $s < s' < s''$. To see why this is impossible, note that: if $B^*(s') = \big\{\{1\}, \{1, 2\}\big\}$, then $r_2(B^*(\,\cdot\,))$ cannot be nondecreasing, contradicting claim (i); if $B^*(s') = \big\{\{2\}, \{1, 2\}\big\}$, then $r_1(B^*(\,\cdot\,))$ cannot be nonincreasing, contradicting claim (ii). 
\end{proof}

\begin{rmk}\label{rmk:demand}
\Cref{prop:rotation} says that if there is a sequence of demand rotations for item $i$, i.e. an increase in the dispersion of consumers' values for item $i$, then the item gets promoted to be the ``upgrade item'' while the other item gets demoted to be the ``basic item'', and the optimal menu first gets coarser and then gets finer. This result complements \citet{johnson2006simple}, who study the effect of value dispersion on a monopolist's quality design. They show that a demand rotation always leads to an expansion of the product line. In contrast, our bundling setting involves the monopolist switching the tiers of different items and adopting a menu size that is $U$-shaped in the dispersion parameter. As an example, consider case (i) of \Cref{ex:two-item}, where $\gamma = 0.5$ and $\beta$ varies from $0$ to $2$. As parameter $\beta$ increases, there is a sequence of demand rotations for item $2$. The optimal menu changes in a way that is consistent with \Cref{prop:rotation}: It shifts from $\big\{\{2\}, \{1, 2\}\big\}$ to $\big\{\{1, 2\}\big\}$, and then to $\big\{\{1\}, \{1, 2\}\big\}$, as parameter $\beta$ increases. 
\end{rmk}

\subsection{Optimal Quality Design and Cost Structures}\label{subsec:supply}

A special case of our model is the quality discrimination model a la \citet{mussa1978monopoly}. Our result provides new insights even in this well-studied setting. Let $\X := \{x_1, \dots, x_n\} \subset \R_{++}$ be a set of qualities, with $x_1 < \cdots < x_n$. In this model, a type-$t$ consumer has value $v_Q(x, t)$ for a good of quality $x$; the monopolist incurs cost $C_Q(x)$ to supply a good of quality $x$. This can be viewed as a special case of our model, where we define the values and costs for the bundles as follows: For all $k = 1, \dots, n$, let 
\[v(\{1, \dots, k\}, t) := v_Q(x_k, t)\,, \quad C(\{1, \dots, k\}) := C_Q(x_k)\,.\]
Let $v(b, t) =  C(b) = 0$ for all bundles $b$ that are not of the form $\{1, \dots, k\}$ (i.e. these bundles can be simply ignored when applying our results). Note that, in this case, the nesting condition is always satisfied. 

With a slight abuse of notation, we drop the subscript $Q$ whenever it is clear. Let $D^*(x)$ be the unique sales volume of the good of quality $x$ when it is sold alone; thus, $D^*: \X \rightarrow [0, 1]$. For simplicity of exposition, assume that $D^*(x) \in (0, 1)$ for all $x \in \X$. Our next result provides a new characterization of optimal quality discrimination:
\begin{prop}\label{prop:D-hat}
Consider the upper decreasing envelope of $D^*$:
\[\widehat{D}^*(x):= \inf\Big\{g(x) : \text{ $g$ is nonincreasing and $g \geq D^*$}\Big\}\,.\]
Let 
\[X^* := \Big\{x: \widehat{D}^*(x) = D^*(x)\Big\}\,.\]
Then selling menu $X^*$ is optimal. 
\end{prop}
\begin{proof}[Proof of \Cref{prop:D-hat}]
By \Cref{thm:po}, it suffices to show that if $x^\dagger \not \in X^*$, then $x^\dagger$ is dominated by another quality level $x \neq x^\dagger$. Fix any $x^\dagger \not \in X^*$. Then $\widehat{D}^*(x^\dagger) > D^*(x^\dagger)$. Suppose, for contradiction, that there does not exist any $x > x^\dagger$ such that $D^*(x) \geq D^*(x^\dagger)$. Then, we have
\[ \max_{x> x^\dagger}\big\{\widehat{D}^*(x)\big\} =  \max_{x> x^\dagger}\big\{D^*(x)\big\} < D^*(x^\dagger)\,.\]
Let
\[\widetilde{D}^*(x)  := \begin{cases} \widehat{D}^*(x) &\text{ if $x \neq x^\dagger$\,;}\\
D^*(x) &\text{ otherwise\,.}
\end{cases}\]
Then note that $\widetilde{D}^*(x)$ is also nonincreasing and everywhere above $D^*(x)$. Moreover, $\widetilde{D}^*(x)$ is everywhere below $\widehat{D}^*(x)$ with $\widetilde{D}^*(x^\dagger) < \widehat{D}^*(x^\dagger)$. A contradiction. 
\end{proof}

\paragraph{Cost Structures.}\hspace{-2mm}Using \Cref{prop:D-hat}, we now characterize how optimal quality design depends on the cost structure. Suppose that the consumers' preferences are multiplicative: $v(x, t) = x \cdot t$.\footnote{This is equivalent to assuming that $v(x, t)= u(x)\cdot v(t)$ for strictly increasing functions $u$ and $v$, by a change of variables.} Let $C_{avg}(x):= C(x) \big / x$ be the average cost function. 
\begin{prop}\label{prop:C-check}
Suppose that $F$ is regular.\footnote{That is, the function $t - \frac{1 - F(t)}{f(t)}$ is strictly increasing (which is equivalent to that the quality-adjusted revenue curve is strictly concave).} Consider the lower increasing envelope of $C_{avg}$:
\[\widecheck{C}_{avg}(x) := \sup \Big\{g(x) : \text{ $g$ is nondecreasing and $g \leq C_{avg}$}\Big\}\,.\]
Let 
\[X^* := \Big\{x: \widecheck{C}_{avg}(x) = C_{avg}(x)\Big\}\,.\]
Then selling menu $X^*$ is optimal.
\end{prop}
\begin{proof}[Proof of \Cref{prop:C-check}]
Note that since $D^*(x) \in (0, 1)$, we must have
\[\text{MR}(D^*(x)) \cdot x - C(x) = 0\,,\]
where the quality-adjusted marginal revenue curve $\text{MR}(q) := \frac{\d}{\d q} (F^{-1}(1-q) \cdot q)$ is a strictly decreasing, continuous function. Therefore, we have 
\[D^*(x) = \text{MR}^{-1}\Big(C_{avg}(x)\Big)\,.\]
Now, observe that for any function $h: \X \rightarrow \R$ and any strictly decreasing function $\Phi: \R \rightarrow \R$, we have 
\[\textbf{U}^-\big[\Phi\circ h\big] = \Phi\circ \textbf{L}^+\big[h\big]\,,\]
where $\textbf{U}^-[\,\cdot\,]$ denotes the upper decreasing envelope operator and $\textbf{L}^+[\,\cdot\,]$ denotes the lower increasing envelope operator. Thus, we have 
\[\widehat{D}^*(x) = \text{MR}^{-1}\Big(\widecheck{C}_{avg}(x)\Big)\]
because $\text{MR}^{-1}(\,\cdot\,)$ is strictly decreasing. The claim follows from \Cref{prop:D-hat}.
\end{proof}

\begin{rmk}\label{rmk:supply}
This result generalizes Proposition 1 of \citet{johnson2003multiproduct}, where the average cost curve is assumed to be $U$-shaped.\footnote{They define a cost structure to be \textit{$U$-shaped} if there exists some quality threshold $x_k$ below which the average cost is decreasing and above which the marginal and average costs are increasing. Note that in this case, the set of undominated qualities $X^*$ coincides with the region of increasing marginal and average costs $\{x_k, \dots, x_n\}$. Moreover, by \Cref{alg:alg1}, the menu $X^*$ in this case is the minimal optimal menu.} They conclude that ``It is optimal to segment the market with multiple products exactly in the region where average cost and marginal cost are increasing'' (\citealt{johnson2003multiproduct}). However, \Cref{prop:C-check} shows that this conclusion is incomplete when the cost structure is more complex.\footnote{Average costs are generally not $U$-shaped whenever there are kinks in the cost function due to a mix of production technologies, e.g. $C(x) = \min\big\{k_1 + x^{\alpha_1}, k_2 + x^{\alpha_2}\big\}$ where $k_1 < k_2$ and $\alpha_1 > \alpha_2$.} The optimal mechanism need not segment the market even when average cost and marginal cost are increasing. \Cref{fig:quality} illustrates. Specifically, the marginal and average costs can both increase within the blue region highlighted in \Cref{fig:c}, yet the optimal mechanism does not segment the market using these qualities. Instead, as illustrated in \Cref{fig:d}, optimal quality choices are characterized by our notion of dominance.
\end{rmk}

\begin{figure}
    \centering
\hspace{-0.08\linewidth}
\begin{subfigure}[b]{0.48\linewidth}
\centering
        \includegraphics[scale=0.58]{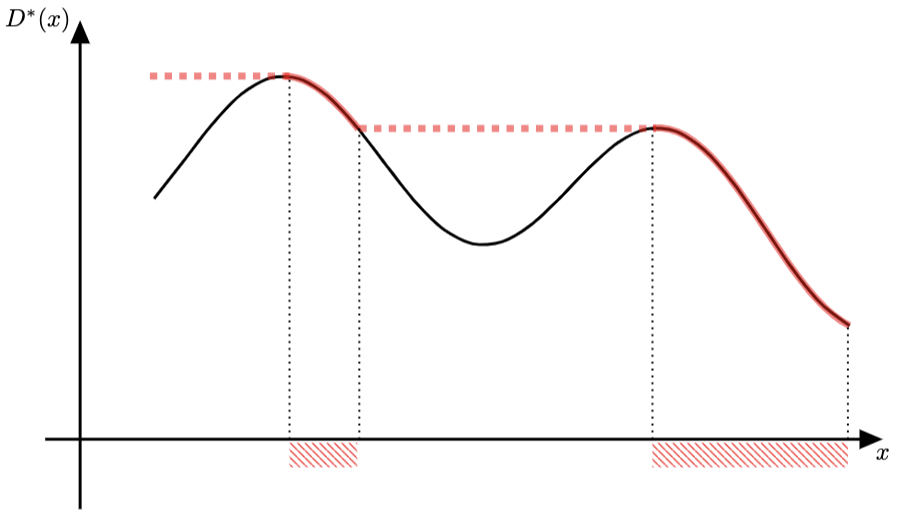}
    \caption{$D^*(x)$ and $\widehat{D}^*(x)$\label{fig:d}}
\end{subfigure}
\hspace{0.06\linewidth}
\begin{subfigure}[b]{0.48\linewidth}
    \centering
        \includegraphics[scale=0.54]{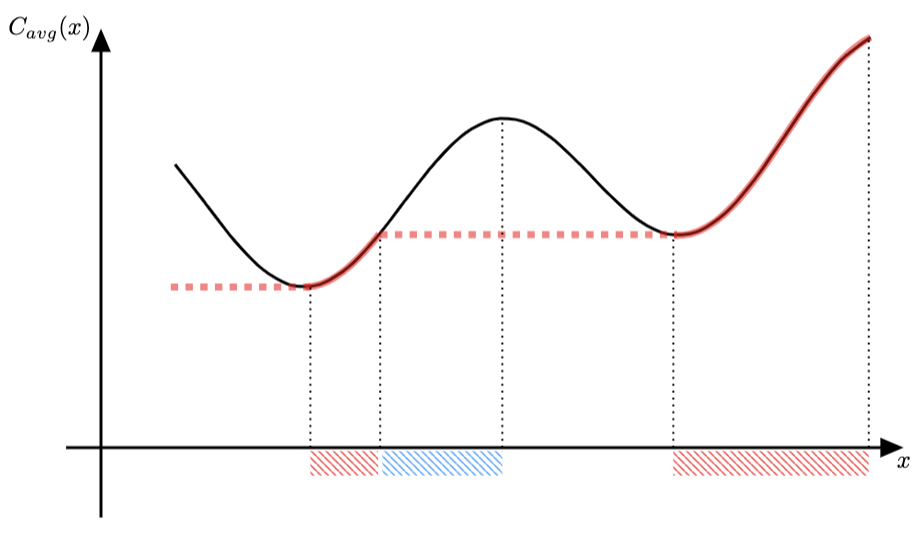}
       \caption{$C_{avg}(x)$ and $\widecheck{C}_{avg}(x)$\label{fig:c}}
\end{subfigure}
\caption{Illustration of the upper decreasing envelope $\widehat{D}^*$ and the lower increasing envelope $\widecheck{C}_{avg}$}
\label{fig:quality}
\end{figure}

\begin{rmk}\label{rmk:bunching}
At first glance, \Cref{prop:D-hat} and \Cref{prop:C-check} may seem related to the ``ironing procedures'' in \citet{mussa1978monopoly} and \citet{Myerson1981}. However, this connection is only superficial. Even though both \Cref{prop:D-hat} and \Cref{prop:C-check} characterize various ``bunching regions'', they operate in a setting where ``ironing'' is not needed. In the standard textbook treatment of one-dimensional screening problems, the regularity assumptions that rule out ``ironing'' also rule out the possibility of ``bunching'' (see pp. 262--268 of \citealt{Fudenberg1991}). Our assumptions are much weaker, allowing for various forms of ``bunching.'' This generality relies on our new characterization of the upper envelope of the MR curves. 
\end{rmk}

\subsection{Optimal Screening with Price and Nonprice Instruments}\label{subsec:screening}
Consider a monopolist selling a set of quality-differentiated goods, as described in \Cref{subsec:supply}. In addition to setting prices for these goods, the monopolist can also use nonprice instruments by requiring customers to perform certain costly actions, such as waiting in line or collecting coupons, in order to qualify for certain offers. When is such nonprice screening optimal? 

We consider a special case of the model introduced in \citet{yang2022costly}. A consumer's payoff is given by \[\underbrace{u(x, t)}_\text{utility of consuming quality $x$} - \underbrace{c(y, t)}_\text{disutility of costly action $y$} - \quad \underbrace{p}_\text{price}\]
where $x \in \{x_1, \dots, x_n\} =: \mathcal{X} \subset \R_{++}$ denotes the quality and $y \in \{y_1, \dots, y_m\} =: \mathcal{Y} \subset \R_{++}$ denotes the costly action, with the normalization $u(0, t) = c(0, t) = 0$. The seller's payoff is given by $-C(x) + p$, where $C(\,\cdot\,)$ is the production cost function. From \citet{yang2022costly}, we know that if $c(y, t)$ is \textit{nonincreasing} in $t$, then the optimal deterministic mechanism does not use the costly instruments (i.e. $y(t) = 0$ for all types $t$).

In this section, we consider the opposite case where $c(y, t)$ is strictly \textit{increasing} in $t$. We also restrict attention to deterministic mechanisms $(x, y, p): T \rightarrow (\{0\} \cup \X) \times (\{0\} \cup \Y) \times \R$. We say that \textit{costly screening is optimal} if every optimal mechanism requires a positive mass of consumers to perform some costly action $y \in \Y$. 

Let $\pi(x, q)$ be the profit function of selling quality $x$ alone and $D^*(x)$ the corresponding sales volume as in \Cref{subsec:supply}. We strengthen the quasi-concavity assumptions in \Cref{subsec:supply} to strict concavity and assume that $D^*(x) < 1$ for all $x \in \X$. To state our result, for any costly action $y \in \Y$, we introduce the following hypothetical profit function: 
\[\pi(y, q) := c\big(y, F^{-1}(1-q)\big) \cdot q\]
which is the profit that the seller would obtain if she were to sell the right to \textit{opt out of} the costly action $y$ to a mass $q$ of individuals. Suppose that $\pi(y, q)$ is strictly concave in $q$. Let $D^*(y)$ be the unique sales volume that maximizes this hypothetical profit function.

Let $\mathcal{S}$ be the set of allocations that involve costly actions but generate a positive total surplus for at least some type $t$:
\[\mathcal{S}:= \big\{(x, y) \in \X \times \Y: u(x, t) - c(y, t) - C(x) > 0 \text{ for some type $t$}\big\}\,.\]
Note that if there exists some $(x, y) \in \mathcal{S}$ such that the net value $u(x, t) - c(y, t)$ is strictly \textit{decreasing} in $t$, then costly screening is optimal.

Now, to consider the more interesting case, suppose that (i) the net value $u(x, t) - c(y, t)$ is strictly increasing in $t$ and (ii) the net profit function $\pi(x, q) - \pi(y, q)$ is strictly quasi-concave in $q$ for all $(x, y) \in \mathcal{S} \cup \{(x^*, y^*)\}$, where $x^* := \max\{\argmax_x D^*(x)\}$ and $y^* := \min \{ \argmin_y D^*(y)\}$. Then, we have the following result:
\begin{prop}\label{prop:screening}
Costly screening is optimal if and only if $\min_y D^*(y) < \max_{x} D^*(x)$. 
\end{prop}
\begin{proof}[Proof of \Cref{prop:screening}] Without loss of generality, we may assume that the seller only uses allocations in $\mathcal{S} \cup \{(x, 0): x \in \X\}$. Indeed, for any option that involves an allocation $(x, y)$ such that the surplus $u(x, t) - C(x) - c(y, t) \leq 0$ for all types $t$, the seller must obtain a non-positive profit (because of the IR constraints). But the seller can always weakly improve by simply removing all options with non-positive profits from the menu.

Now, we map this problem into a bundling problem as follows. Consider a bundling problem with $n + m$ many items, where the first $n$ items represent \textit{quality upgrades} exactly as in \Cref{subsec:supply}, and the remaining $m$ items represent \textit{opting out of} each of the $m$ costly activities. Specifically, for any $(x_i, y_j) \in \mathcal{S}$, we define 
\[v\Big(\big\{1, \dots, i\big\}\cup \{n+1, \dots, n+m\big\} \backslash \big\{n+j\big\} , t\Big) := u(x_i, t) - c(y_j, t)\,,\]
with $v\big(\big\{1, \dots, i\big\}\cup \{n+1, \dots, n+m\big\} , t\big) := u(x_i, t)$ being the value of quality $x_i$ without any costly action. We can map the production costs accordingly and let $v(b, t) = C(b) = 0$ for bundles $b$ that are not of the above form (i.e. ignore them when applying our results). One can verify that this bundling problem satisfies our assumptions. With a slight abuse of notation, we also write $(x, y)$ as the bundle of quality $x$ and costly action $y$, and write $D^*(x, y)$ as the corresponding sales volume of this ''damaged good'' when sold alone. 

\textbf{The ``if'' part:} Suppose that $\min_y D^*(y) < \max_{x} D^*(x)$. Suppose for contradiction that there exists an optimal deterministic mechanism that does not use any costly instruments. Then, by \Cref{prop:D-hat}, there exists an optimal, nested menu $B$ such that $x^* = \max\big\{\argmax_x D^*(x)\big\}$ is the smallest non-empty bundle in $B$. Let $y^* = \min\big\{\argmin_y  D^*(y)\big\}$. By assumption, $D^*(y^*) < D^*(x^*)$. Because $\pi(x^*, q)$, $\pi(y^*, q)$, and $\pi(x^*, q) - \pi(y^*, q)$ are strictly quasi-concave, this implies that
\[ D^*(y^*) < D^*(x^*)  < D^*({x^*, y^*}) \,.\]
So we have that $(x^*, y^*) \in \mathcal{S}$ and that the bundle $(x^*, y^*)$ sells more than the bundle $x^*$ does when both are sold alone. Since the bundle $(x^*, y^*)$ is a subset of the bundle $x^*$, this implies that the menu $B$ cannot be optimal by \Cref{thm:char}.\footnote{The necessity part of \Cref{thm:char} involves adding probabilistic assignments but the assignments can be made deterministic if the bundle $b_1 \not\in B$ is a subset of the smallest non-empty bundle in menu $B$.} A contradiction.

\textbf{The ``only if'' part:} Suppose that $\min_y D^*(y) \geq \max_{x} D^*(x)$. Then, for all $(x', y') \in \mathcal{S}$,
\[D^*(y') \geq \min_y D^*(y) \geq \max_{x} D^*(x) \geq D^*(x')\,.\]
Because $\pi(x', q)$, $\pi(y', q)$, and $\pi(x', q) - \pi(y', q)$ are strictly quasi-concave, this implies that
\[D^*(y')  \geq D^*(x') \geq D^*(x', y')\,.\]
So the bundle $(x', y')$ is dominated by the bundle $x'$. Eliminating these dominated bundles leaves a nested menu consisting of only bundles that involve no costly action. By \Cref{thm:po}, the resulting menu is optimal (among all stochastic mechanisms and hence also optimal among all deterministic mechanisms). Thus, costly screening is not optimal. 
\end{proof}

\begin{rmk}
Note that our model does not require $c(y, t)$ to be monotone in $y$, nor does it impose an increasing differences condition on $c(y, t)$.
Thus, this model allows different consumers to have different ordinal rankings for the costly actions. In addition, as the proof shows, for the ``if'' part of \Cref{prop:screening}, it suffices to assume the monotonicity of $u(x, t) - c(y, t)$ and the quasi-concavity of $\pi(x, q) - \pi(y, q)$ for a single pair $(x^*, y^*)$. 
\end{rmk}

\begin{rmk}\label{rmk:costela}
To interpret \Cref{prop:screening}, recall that 
\[D^*(y) = \argmax_q \Big\{c\big(y, F^{-1}(1-q)\big) \cdot q\Big\}\,.\]
Note that $D^*(y)$ is smaller if the elasticity of disutility $c(y, t)$ with respect to type $t$ is pointwise higher (i.e. the costly action is differentially more costly for a higher type). That is, we have 
\[\frac{\d \log c(y_1, t)}{\d \log t} \leq \frac{\d \log c(y_2, t)}{\d \log t} \, \text{ for all $t$ } \implies  D^*(y_1) \geq D^*(y_2) \,.\]
\Cref{prop:screening} shows that, under our assumptions, costly screening is optimal if and only if there exists a costly action $y$ with sufficiently high elasticity of disutility in a precise sense: $D^*(y)$ is smaller than $\max_{x} D^*(x)$. \Cref{fig:screening} illustrates. 

The intuition behind this characterization can be understood from our main results. In the absence of costly screening, by \Cref{prop:D-hat}, we know that the best-selling quality $x^* := \max\{\argmax_x D^*(x)\}$ would be optimally offered as the base quality level in the menu. If there exists a costly action $y$ such that $D^*(y) < D^*(x^*)$, then the ``damaged good'' $(x^*, y)$ --- which requires action $y$ to purchase quality $x^*$ --- has an even higher sales volume. Intuitively, this is because the costly action $y$ compresses the distribution of values for the ``damaged good.'' Thus, the ``damaged good'' $(x^*, y)$ can be profitably included in the menu to expand the market by \Cref{thm:char}. On the other hand, if $D^*(y) \geq D^*(x^*)$, then we have $D^*(y) \geq D^*(x^*) \geq D^*(x)$ for all qualities $x$, since quality $x^*$ is the best-selling quality. This implies that any ``damaged good'' $(x, y)$ has a lower sales volume than its ``undamaged version'' $x$. When mapped into our bundling problem, this means that every such bundle is dominated by some bundle involving no costly action. The bundles involving no costly action are nested by construction because they represent different quality levels, and thus, costly screening is suboptimal by \Cref{thm:po}.
\end{rmk}

\begin{figure}
    \centering
    \includegraphics[scale=0.62]{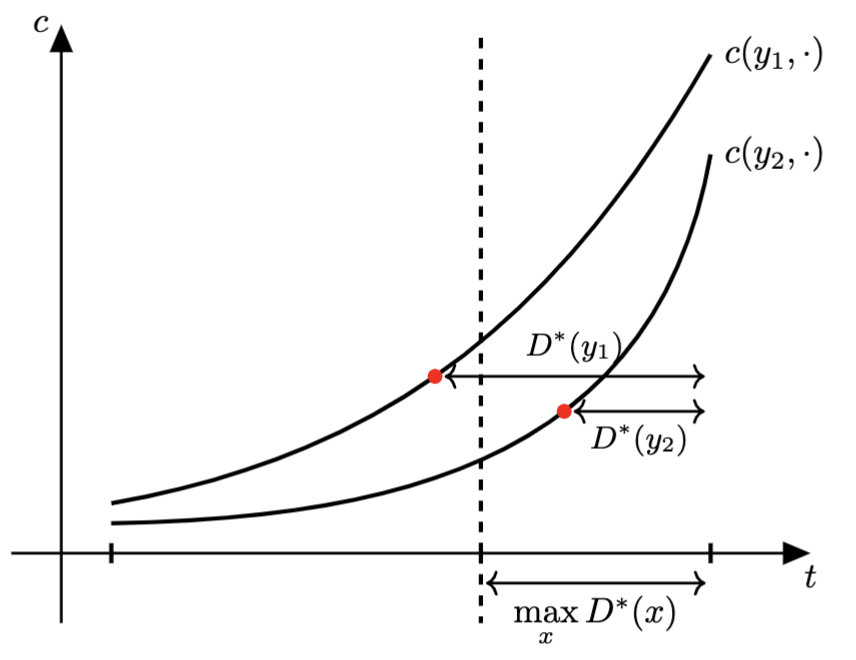}
    \caption{Illustration of \Cref{prop:screening}}
    \label{fig:screening}
\end{figure}

\section{Conclusion}\label{sec:conclusion}

This paper studies optimal bundling when consumers differ in one dimension. We introduce a partial order on the set of bundles defined by (i) set inclusion and (ii) sales volumes. This simple partial order turns out to characterize optimal bundling strategies. We show that if the set of undominated bundles is nested, then nested bundling is optimal. In particular, selling the set of undominated bundles is optimal. The proof uses a novel characterization of the upper envelope of the MR curves. 

We apply these insights to connect optimal bundling and quality design to price elasticities and cost structures. We also show how these insights can be applied to other multidimensional mechanism design problems. In particular, we obtain a necessary and sufficient condition for the optimality of costly screening when the principal has access to both price and nonprice screening instruments.

\setlength\bibsep{12pt}
\bibliographystyle{ecta} 
\newpage
\bibliography{references}

\newpage
\appendix
\crefalias{section}{appendix}

\section{Proof of the Main Result}\label{sec:proof}

This appendix provides the proof of the main result, following the proof sketch outlined in \Cref{sec:sketch}.

\subsection{Proof of \texorpdfstring{\Cref{thm:po}}{}}\label{subsec:proof-po}
To begin, we explicitly construct a menu and its associated prices as follows:

\vspace{0.5cm}

\hspace{-0.6cm}\fbox{\begin{minipage}{38em}
\begin{alg}[Minimal optimal menu]\label{alg:alg1} \text{ }
\begin{itemize}
    \item[1.] Let $U$ be the set of undominated bundles (which by assumption are nested). 
    \item[2.] Order the bundles in $U$ such that $U = \{\emptyset, b_1,\dots, b_m\}$ where $b_i \subset b_j$ for all $i < j$. 
    \item[3.] Initialize $B = \{\emptyset\}$,  $q(b) = 1$ for all $b \in U$, and $i = 1$.
    \item[4.] While $i \leq m$:
        \begin{itemize}
            \item[(i)]  Let $\hat{b}$ be the largest bundle in $B$. 
            \item[(ii)] Optimize quantity $q^* \in [0, q(\hat{b})]$ for incremental profit $\pi(b_i, q) - \pi(\hat{b}, q)$
            \item[(iii)] If $q^* = q(\hat{b})$: Update $B = B \backslash \{\hat{b}\}$.
            \item[(iv)] If $q^* = 0$: Update $i = i + 1$.
            \item[(v)]  If $0 < q^* < q(\hat{b})$ or $B = \emptyset$: Update $B = B \cup \{b_i\}$, $q(b_i) = q^*$, and $i = i + 1$.
        \end{itemize}
    \item[5.] Return $B$ and $\{q(b)\}_{b\in B}$. 
\end{itemize}
\quad \textbf{Prices:} Order the bundles in menu $B$ such that $B = \{\emptyset, b_1, \dots, b_l\}$ where $b_1 \subset \cdots \subset b_l$. We construct the corresponding prices using upgrade prices. Specifically, the upgrade price from bundle $b_i$ to bundle $b_{i+1}$ is set such that, when all consumers have bundle $b_i$ and no other bundles are available in the market, the quantity of the incremental bundle $b_{i+1} \backslash b_{i}$ sold equals $q(b_{i+1})$.
\end{alg}
\end{minipage}
}

\vspace{0.5cm}

In the remainder of the proof, we show that: (i) selling the constructed menu $B$ with the associated prices is optimal among all stochastic mechanisms, and (ii) the menu $B$ is minimally optimal. We defer the proof of the uniqueness part of the statement to the end.

Following \citet{Myerson1981}, let 
\[\phi(b, t) := v(b, t) - C(b) - \frac{1 - F(t)}{f(t)} v_t(b, t)\]
be the \textit{virtual surplus} function for bundle $b$. As observed by  \citet{bulow1989simple}, this can be interpreted as the \textit{marginal profit} for bundle $b$ evaluated at the quantity such that the marginal consumer is of type $t$. 

\begin{lemma}\label{lem:relaxed}
Let $a^*$ be a solution to 
\[\sup_{a:\, T \rightarrow \Delta(\mathcal{B})} \E\Bigg [\sum_b a_b(t) \phi(b, t) \Bigg]\,. \label{eq:relaxed}\]
If there exists a pricing rule $p^*$ such that (i) $(a^*, p^*)$ is a mechanism and (ii) type $\underline{t}$ gets zero payoff under $(a^*, p^*)$, then $(a^*, p^*)$ is optimal. 
\end{lemma}
\begin{proof}[Proof of \Cref{lem:relaxed}]
By the envelope theorem, under any mechanism $(a, p)$ that gives type $\underline{t}$ zero payoff, the profit to the seller must be $\E [\sum_b a_b(t) \phi(b, t) ]$ (see \citealt{Myerson1981}). The claim follows. 
\end{proof}

\paragraph{Step 1.}\hspace{-2mm}Let $t^*(b) := F^{-1}(1 - D^*(b))$ be the unique cutoff type for bundle $b$, which is where $\phi(b, t)$ crosses $0$ from below. Clearly, 
\[D^*(b_1) \geq D^*(b_2) \iff t^*(b_1) \leq t^*(b_2)
\,.\]

A function $g(t)$ \textit{strictly single-crosses (from below)} another function $h(t)$ on an interval $[t_1, t_2]$ if $g(t) \geq h(t) \implies g(t') > h(t')$ for all $t < t' \in [t_1, t_2]$.  

\begin{lemma}\label{lem:cross1}
For any $b_1 \subset b_2$, $\phi(b_2, t)$ strictly single-crosses $\phi(b_1, t)$ on $[\max(t^*(b_1), t^*(b_2)), \overline{t}]$.
\end{lemma}
\begin{proof}[Proof of \Cref{lem:cross1}]
Observe that    
\[\phi(b, t) = \frac{\d \pi(b, q)}{\d q}\bigg |_{\,\,q = 1- F(t)} \,.\]
By assumption, $\pi(b_2, q) - \pi(b_1, q)$ is strictly quasi-concave on $[0,\min(D^*(b_1), D^*(b_2))]$ and has at most one stationary point in $[0, \min(D^*(b_1), D^*(b_2))]$ which if exists is the maximum point in $[0, \min(D^*(b_1), D^*(b_2))]$. The claim follows immediately.
\end{proof}

For any $b_1\subset b_2$, let 
\[s(b_1, b_2) := \inf_{s}\Big \{s \in T: \phi(b_2, t) > \phi(b_1, t) \text{ for all $t > s$}\Big\}\]
be the \textit{last crossing type} of the virtual surplus functions for $b_1$ and $b_2$. Note that $s(b_1, b_2)$ is always well-defined. Let
\[\chi(b_1, b_2) := \phi(b_1, s(b_1, b_2))\]
be the \textit{last crossing virtual surplus} of $b_1$ and $b_2$. The next lemma shows that the sign of the last crossing virtual surplus is determined by the ranking of the sales volumes:

\begin{lemma}[Virtual surplus and sales volumes]\label{lem:cross2}
For any $b_1 \subset b_2$, 
\[\sign\Big[\chi(b_1, b_2)\Big]  = \sign\Big[D^*(b_1) - D^*(b_2) \Big]\,.\]
\end{lemma}
\begin{proof}[Proof of \Cref{lem:cross2}]
Observe that the claim holds if either $s(b_1, b_2) = \underline{t}$ or $s(b_1, b_2) = \overline{t}$. Now suppose that $s(b_1, b_2) \in (\underline{t}, \overline{t})$. Observe that by \Cref{lem:cross1}, only the following three cases can happen:
\begin{itemize}
    \item[] Case (1): $s(b_1, b_2) > \max(t^*(b_1), t^*(b_2))$ and $t^*(b_1) < t^*(b_2)$\,;
    \item[] Case (2): $s(b_1, b_2) < \min(t^*(b_1), t^*(b_2))$ and $t^*(b_1) > t^*(b_2)$\,;
    \item[] Case (3): $s(b_1, b_2) = \min(t^*(b_1), t^*(b_2)) = \max(t^*(b_1), t^*(b_2))$\,.
\end{itemize}
This argument is depicted in \Cref{fig:crossing} (in \Cref{subsec:sketch-po}) in which we hold the relative position of $\phi(b_1, t)$ and $\phi(b_2, t)$ fixed and vary the $x$-axis. The claim follows immediately. 
\end{proof}

\paragraph{Step 2.}\hspace{-2mm}Let $U$ be the set of undominated bundles. We now show that any dominated bundle can be eliminated in the following sense:

\begin{lemma}[Elimination of dominated bundles]\label{lem:remove} We have 
\[\sup_{a:\, T \rightarrow \Delta(\mathcal{B})} \E\Big [\sum_b a_b(t) \phi(b, t) \Big] = \sup_{a:\, T \rightarrow U} \E\Big [\sum_b a_b(t) \phi(b, t) \Big]\,.\] 
\end{lemma}
\begin{proof}[Proof of \Cref{lem:remove}]
Note that by compactness and linearity, a solution to the left-hand side exists and uses only deterministic assignments. Let $a$ be any such a solution. Suppose that there exists a dominated bundle $b$ such that $a(t) = b$ for some $t \in T$ (otherwise the claim holds). First, observe that because the dominance relation $\preceq$ is a partial order on a finite set, any dominated bundle must also be dominated by some undominated bundle. So there exists $b' \in U$ that dominates $b$. 

Second, observe that by \Cref{lem:cross2}, we have $\chi(b, b') \leq 0$, and therefore $\phi(b', t) \geq \phi(b, t)$ for all $t \geq t^*(b)$, which implies that \[\max\big\{0,\phi(b', t) \big\} \geq \phi(b, t)\]
for all $t \in T$.  Therefore, by replacing $a(t)$ from $b$ to either $\emptyset$ or $b'$ for all types that are previously assigned $b$, we can weakly improve the objective value. Since this argument holds for all dominated bundles $b$, the claim follows. 
\end{proof}

\paragraph{Step 3.}\hspace{-2mm}Now, we construct a solution to 
\[\sup_{a:\, T \rightarrow U} \E\Bigg [\sum_b a_b(t) \Big(v(b, t) - C(b) - \frac{1 - F(t)}{f(t)} v_t(b, t) \Big) \Bigg]\]
and show that it coincides with the allocation rule induced by the menu $B$ and associated prices constructed by \Cref{alg:alg1}. Because the set of undominated bundles is nested, we can write $U = \{\emptyset, b_1, \dots, b_m\}$ with $b_1 \subset \cdots \subset b_m$. Moreover, we must also have 
\[D^*(b_1) > D^*(b_2) > \cdots > D^*(b_m)\,.\]
For all $j = 1, \dots, m$, let
\[V_j(t) := \max\Big\{0, \phi(b_1, t), \dots, \phi(b_j, t)\Big\}\]
be the \textit{envelope} of the virtual surplus functions of the $j$ smallest bundles (and the empty bundle). 

The next lemma shows that taking envelopes preserves single-crossing properties: 
\begin{lemma}[Single-crossing envelopes]\label{lem:envelope}
$\phi(b_{k}, t)$ strictly single-crosses $V_j(t)$ for all $j < k$.
\end{lemma}
\begin{proof}[Proof of \Cref{lem:envelope}]
We prove this by induction on $j$. 

\textbf{Base case:} Let $j = 1$. Fix any $k > 1$. Note that $b_1 \subset b_k$ and $D^*(b_1) > D^*(b_k)$. Therefore, by \Cref{lem:cross1}, we have that $\phi(b_k, t)$ strictly single-crosses $\phi(b_1, t)$ on $[t^*(b_k), \overline{t}]$, which implies that $\phi(b_k, t)$ strictly single-crosses $\max\{0, \phi(b_1, t)\} = V_1(t)$ on $[\underline{t}, \overline{t}]$.

\textbf{Inductive step:} Suppose that for all $k > j - 1$, $\phi(b_k, t)$ strictly single-crosses $V_{j-1}(t)$. Fix any $k > j$. We can write
\[V_j(t) = \max \Big \{V_{j-1}(t), \max\big\{0, \phi(b_j, t)\big\} \Big \}\,.\]
Note that
\begin{itemize}
    \item[] (i) $\phi(b_k, t)$ strictly single-crosses $\max\{0, \phi(b_j, t)\}$ (by the same argument as before);
    \item[] (ii) $\phi(b_j, t)$  strictly single-crosses $V_{j-1}(t)$ (by the inductive hypothesis);
    \item[] (iii) $\phi(b_k, t)$ strictly single-crosses $V_{j-1}(t)$ (by the inductive hypothesis).
\end{itemize}
Note that if any of the above crossings does not happen, then the claim for the inductive step clearly holds. Otherwise, let $t(k, j)$, $t(j, j-1)$, and $t(k, j-1)$ be the crossing points for (i), (ii), and (iii), respectively. 

Now, because of the above single-crossing properties, observe that only the following three cases can happen:
\begin{itemize}
    \item[] Case (1): $t(k, j-1) < t(j, j-1)$ in which case we must have $t(k, j) < t(k, j-1)$. 
    \item[] Case (2): $t(k, j-1) > t(j, j-1)$ in which case we must have $t(k, j) > t(k, j-1)$.  
    \item[] Case (3): $t(k, j-1) = t(j, j-1)$ in which case we must have $t(k, j) = t(k, j-1)$.   
\end{itemize}
This argument is depicted in \Cref{fig:envelope} (in \Cref{subsec:sketch-po}) in which we hold the positions of $V_{j-1}(t)$ and $\phi(b_j, t)$ fixed and vary $\phi(b_k, t)$. The claim for the inductive step follows.  
\end{proof}

We now iteratively construct the following allocation rule: 
\begin{lemma}[Iterative construction]\label{lem:alg}
There exists a mapping $a: T \rightarrow U$ such that 
\begin{itemize}
    \item[(i)] $\phi(a(t), t) = V_m(t)$ for all $t$.
    \item[(ii)] $a(t) \subseteq a(t')$ for all $t < t'$. 
\end{itemize}
\end{lemma}
\begin{proof}[Proof of \Cref{lem:alg}]
Let $a^{(0)}(t) = \emptyset$. For each $i = 1, \dots, m$:
\begin{itemize}
    \item[1.] Let $t(i, i-1)$ be the type at which $\phi(b_i, t)$ strictly single-crosses $V_{i-1}(t)$ (by \Cref{lem:envelope}). Put $t(i, i-1) = \overline{t}$ if the crossing does not exist.
    \item[2.] Let 
    \[a^{(i)}(t) =\begin{cases}
    a^{(i-1)}(t) &\text{ if $t \leq t(i, i-1)$\,;} \\
    b_i     &\text{ otherwise}\,.
    \end{cases}\]    
\end{itemize} 
Let $a(t) = a^{(m)}(t)$. Note that by induction, at each iteration $i$, we have for all $t$ \[\phi(a^{(i)}(t), t) = \max\Big\{V_{i-1}(t),\, \phi(b_i, t)\Big\} = V_{i}(t)\,;\]
and therefore for all $t$
\[\phi(a(t), t) = V_m(t)\,,\]
proving part (i). Note also that, by induction, at each iteration $i$, we have for all $t < t'$
\[a^{(i)}(t) \subseteq a^{(i)}(t')\]
because $b \subset b_i$ for any $b \in \{a^{(i-1)}(t)\}_{t \in T}$. So $a(t) \subseteq a(t')$ for all $t < t'$, proving part (ii). 
\end{proof}

Let $a^\dagger$ be the allocation rule constructed in \Cref{lem:alg}. Order the bundles such that $\big\{a^\dagger(t)\big\}_{t \in T} = \{\emptyset, b_1, \dots, b_l\}$, where $b_1 \subset \dots \subset b_l$, for some number $l$. For all $j=1,\dots,l$,
let $s(b_j) = \inf \big \{t: a^\dagger(t) = b_j \big \}$. For all $j=1,\dots,l$, let 
\[p^\dagger(b_j) = v(b_j, s(b_j)) - \sum_{i < j} \big(v(b_i, s(b_{i+1})) - v(b_i, s(b_i))\big)\,.\]
Because (i) the allocation rule $a^\dagger$ is monotone in the set-inclusion order, and (ii) the incremental value is monotone for any pair of nested bundles, observe that the prices $p^\dagger$ implement the allocation rule $a^\dagger$.

Now, by the construction of $a^\dagger$ in \Cref{lem:alg}, observe that (i) the set of assigned bundles $\big\{a^\dagger(t)\big\}_{t \in T}$ is exactly equal to the menu $B$ constructed by \Cref{alg:alg1}, (ii) we have that $1 - F(s(b)) = q(b)$ for all $b \in B$, where the quantity $q(b)$ is constructed in \Cref{alg:alg1}, and (iii) the price of upgrading from $b_{i}$ to $b_{i+1}$ is given by 
\[p^\dagger(b_{i+1})- p^\dagger(b_{i}) =  v(b_{i+1}, s(b_{i+1})) - v(b_i, s(b_{i+1}))\,, \]
coinciding with the upgrade price from $b_{i}$ to $b_{i+1}$ constructed by \Cref{alg:alg1}. Therefore, $a^\dagger$ is also the allocation rule induced by the menu and prices given by \Cref{alg:alg1}. 

\paragraph{Completion of the Proof of the Optimality Part.}\hspace{-2mm}We have 
\begin{align*}
     \E\Big[\phi(a^\dagger(t), t) \Big] & =  \sup_{a:\, T \rightarrow U} \E\Big [\sum_b a_b(t) \phi(b, t) \Big] \tag{by \Cref{lem:alg}}\\
    &= \sup_{a:\, T \rightarrow \Delta(\mathcal{B})} \E\Big [\sum_b a_b(t) \phi(b, t) \Big] \tag{by \Cref{lem:remove}}\,.
\end{align*}
Therefore, by \Cref{lem:relaxed}, selling the nested menu $B \subseteq U$ given by \Cref{alg:alg1} is optimal among all stochastic mechanisms. In addition, dropping any bundle $b$ from the menu $B$ would result in a strictly suboptimal menu because, by the proof of \Cref{lem:alg}, for any $b \in B$, we have 
\begin{align*}
\E\Big[\max_{b' \in B \backslash \{b\}} \phi(b', t) \Big] < \E\Big[\max_{b' \in U} \phi(b', t) \Big] = \E\Big[\phi(a^\dagger(t), t) \Big]\,.
\end{align*}
Therefore, the menu $B$ is minimally optimal.

\paragraph{Proof of the Uniqueness Part.}\hspace{-2mm}We now show that any optimal mechanism $(a', p')$ must be equivalent to the nested bundling mechanism constructed by \Cref{alg:alg1}. Suppose, for contradiction, that $(a', p')$ differs from our construction on a positive-measure set of types. By the payoff equivalence theorem as in \Cref{lem:relaxed}, it must be that the allocation rule $a'$ differs from the allocation rule $a^\dagger$ (constructed in \Cref{lem:alg}) on a positive-measure set of types. Let $T^*$ be that set. By the proofs of \Cref{lem:remove} and \Cref{lem:alg}, we have that for all types $t$, 
\[\sum_b a'_{b}(t) \phi(b, t) \leq \max_{b \in U} \phi(b, t) =  \sum_b a^\dagger_{b}(t) \phi(b, t) \,.\]
Let $B :=  \{a^\dagger(t)\}_{t \in T}$. Order the bundles in $B$ such that $B = \{\emptyset, b_1, \dots, b_l\}$ where $b_1 \subset \dots \subset b_l$. For all $b \in B$, let $T_{b} := \{t: a^\dagger(t) = b\}$. Note that there must exist some $b^\dagger \in B$ such that 
\[\P \big(t \in T^* \cap T_{b^\dagger}\big) > 0\,.\]
Note also that by the construction of the allocation rule $a^\dagger$, for all types $t < \sup\{t: t\in T_{\emptyset}\} = t^*(b_1)$ and all bundles $b \neq \emptyset$, we have $\phi(b, t) < 0$. Thus, if $b^\dagger = \emptyset$, then $(a', p')$ must be strictly suboptimal, which is impossible. Now suppose $b^\dagger \neq \emptyset$. By the same argument as in the proof of \Cref{lem:remove}, for any dominated bundle $b'$, there exists some undominated $b''$ such that
\[\phi(b', t) < \max\big\{0, \phi(b'', t)\big\} \]
for almost all $t$. Now, note that by the construction of $a^\dagger$, for any undominated bundle $b'' \neq b^\dagger$, we have 
\[ \max\big\{0, \phi(b'', t)\big\} < \phi(b^\dagger, t)\]
for almost all $t \in T_{b^\dagger}$. These two observations imply that for any bundle $b \neq b^\dagger$, we have 
\[\phi(b, t) < \phi(b^\dagger, t)\]
for almost all $t \in T_{b^\dagger}$. Then $(a', p')$ must be strictly suboptimal. A contradiction. 

\subsection{Proof of \texorpdfstring{\Cref{thm:char}}{}}\label{subsec:proof-char}

\paragraph{Proof of the Sufficiency Part.}\hspace{-2mm}Note that this part follows directly from the proof of \Cref{thm:po}. Now, we further show that this sufficiency part holds even if we only impose monotone incremental values assumption and locally quasi-concave incremental profits assumption on the set of pairs $b_1 \subset b_2$ where $b_2 \in B$. First, both \Cref{lem:cross1} and \Cref{lem:cross2} still hold for every pair $b_1 \subset b_2$ where $b_2 \in B$. Second, note that for every $b_1 \not \in B$, there exists some $b_2 \in B$ such that (i) $b_1 \subset b_2$ and (ii) $D^*(b_1) \leq D^*(b_2)$. Therefore, by the same proof as in \Cref{lem:remove}, we have that it suffices to consider the allocation rules $a: T \rightarrow B$ for the purpose of solving the relaxed problem. Third, note that for all $b_1, b_2 \in B$ such that $b_1 \subset b_2$, we have that $\chi(b_1, b_2) > 0$ (by \Cref{lem:cross2}). Therefore, by the same proof as in Step 3 of the previous section (\Cref{lem:envelope} and \Cref{lem:alg}), we have that there exists an allocation rule $a: T \rightarrow B$ such that (i) $a$ solves the relaxed problem and (ii) $a$ is monotone (in the set-inclusion order). Because the increasing differences property between $v(b_1, t)$ and $v(b_2, t)$ still holds for all $b_1 \subset b_2$ such that $b_1, b_2 \in B$, we have that the allocation rule $a$ is implementable, proving the result by \Cref{lem:relaxed}. The prices can be constructed in the same way as before.

\paragraph{Proof of the Necessity Part.}\hspace{-2mm}Now, we prove the second statement of \Cref{thm:char} by contradiction. Suppose that selling menu $B$ is optimal and that selling any proper subset of menu $B$ is not optimal. 

Suppose, for contradiction, that either 
\begin{itemize}
    \item[(i)] there exist $b_1 \in B$ and $b_2 \in B$ such that $b_1 \subset b_2$ and $D^*(b_1) \leq D^*(b_2)$, or 
    \item[(ii)] there exists $b_1 \not \in B$ such that $D^*(b_1) > D^*(b_2)$ for all non-empty $b_2 \in B$.
\end{itemize}

\paragraph{Step 1.}\hspace{-2mm}Let 
\[U_B := \Big\{b \in B: \text{$b$ is not dominated by another bundle $b' \in B$}\Big\}\,.\]
Note that for all $b_1 \in B \backslash U_B$, there exists $b_2 \in U_B$ such that for all $t$, 
\[\phi(b_1, t) \leq \max\big\{0, \phi(b_2, t)\big\}\,,\]
by the same argument in Step 2 of the previous section. Therefore, by the same arguments in Step 2 and Step 3 of the previous section, there exists a mechanism $(a, p)$ such that $a: T \rightarrow U_B$ and $(a, p)$ attains a profit of 
\[\E\Big[\max_{b \in B} \big\{\phi(b, t)\big\}\Big]\,,\]
which is an upper bound on the profit of selling menu $B$. Because selling $B$ is optimal, this implies that the mechanism $(a, p)$ must be an optimal mechanism. 

\paragraph{Step 2.}\hspace{-2mm}First, consider Case (i). Note that under Case (i), $U_B$ is a strict subset of $B$. Then, because $(a, p)$ is an optimal mechanism, we have that selling the menu $U_B$ is also optimal, violating the minimal optimality of the menu $B$. Hence, this case cannot occur. 

\paragraph{Step 3.}\hspace{-2mm}Now, consider Case (ii). We construct a strict improvement. Consider offering the following new option in addition to the menu given by $(a, p)$: A lottery of getting bundle $b_1$ with probability $\epsilon$, for a price of $\epsilon \times v(b_1, t^*(b_1))$.

We compute the payoff to the monopolist after the addition of this new option. By the payoff equivalence theorem as in \Cref{lem:relaxed}, the profit change to the monopolist after adding this new option can be computed as
\[\E\Bigg[\sum_{b} a'_b(t) \phi(b, t)\Bigg] - \E\Bigg[\sum_{b} a_b(t) \phi(b, t)\Bigg]\,,\]
where $a'$ is the induced allocation rule after the types re-adjust their optimal choices given the new option. 

Let $\underline{b}$ be the smallest non-empty bundle in $B$. By the construction of $(a, p)$, note that: 
\begin{itemize}
    \item[(i)] all types $t \in [\underline{t}, t^*(b_1))$ will not take this option\,;
    \item[(ii)] all types $t \in [t^*(b_1), t^*(\underline{b}))$ will switch from $\emptyset$ to this option\,;
    \item[(iii)] all types $t \in [t^*(\underline{b}), t^\dagger)$ will switch from $\underline{b}$ to this option, for some threshold $t^\dagger$\,.
\end{itemize}

First, we provide a lower bound on the gain in the virtual surplus. Because types in $[t^*(b_1), t^*(\underline{b}))$ will take this new option, the gain in the virtual surplus is at least 
\[\text{Gain}(\varepsilon) := \epsilon \times \underbrace{\int_{t^*(b_1)}^{t^*(\underline{b})} \phi(b_1, t) \d F(t)}_{=:K} = \epsilon K >0\,,\]
where the inequality $K > 0$ uses the single-crossing property of $\phi(b_1, t)$.

Now, we provide an upper bound on the loss in the virtual surplus. Note that any type $t$ who takes this option obtains a payoff that is at most 
\[h(\epsilon) := \epsilon \times \underbrace{\big (v(b_1, \overline{t}) - v(b_1, t^*(b_1)) \big)}_{=:Z} = \epsilon Z\,.\]
Let $b'$ be the second smallest non-empty bundle in $B$ (if it does not exist, put $t^*(b') = 1$ in what follows). Note that $t^*(b_1) < t^*(\underline{b}) < t^*(b')$. By the construction of $(a, p)$, for any $\delta \in [0, t^*(b') - t^*(\underline{b})]$, we have
\[U(t^*(\underline{b}) + \delta) = v(\underline{b}, t^*(\underline{b}) + \delta) - v(\underline{b}, t^*(\underline{b}))\,,\]
where $U$ denotes the indirect utility function under $(a, p)$. Let $g(\delta) := U(t^{*}(\underline{b}) + \delta)$. Note that $v(\underline{b}, t^*(\underline{b})) > 0$ and hence $v_t(\underline{b}, t^*(\underline{b})) > 0$ by assumption. Thus, $\partial_{+} g(0) > 0$. Since $g'$ is continuous on $[0, \delta]$, there exist some constants $\overline{\delta} \in (0, t^*(b') - t^*(\underline{b}))$ and $M > 0$ such that $g'(\delta) \geq M$ for all $\delta \in [0, \overline{\delta}]$. Let $\overline{\epsilon} := g(\overline{\delta}) > 0$. Note that for all $ \epsilon \in (0,  \overline{\epsilon})$, we have
\[g^{-1}(\epsilon) = \int_0^{\epsilon} (g^{-1})'(s) \d s = \int_0^{\epsilon}\frac{1}{g'(g^{-1}(s))} \d s \leq \frac{1}{M} \epsilon\,.\]
Note also that any type $t \in [t^*(\underline{b}), \overline{t}]$ switches to this new option only if 
\[U(t) \leq h(\epsilon)\,.\]
Let 
\[\delta(\epsilon):= g^{-1}(h(\epsilon))\,.\]
Then, observe that for all $\epsilon \in (0, \frac{1}{Z}\overline{\epsilon}) $, the loss in the virtual surplus is at most
\begin{align*}
\text{Loss}(\epsilon) 
&:= \int_{t^*(\underline{b})}^{t^*(\underline{b}) + \delta(\epsilon)}\phi(\underline{b}, t) f(t) \d t \\
\\ &\leq \delta(\epsilon) \times \underbrace{\max_{t \in [t^*(\underline{b}), t^*(\underline{b}) + \delta(\epsilon)]} \Big\{f(t) \phi(\underline{b}, t)\Big\}}_{=:\Phi(\epsilon)} = \delta(\epsilon) \times \Phi(\epsilon) \leq \frac{Z}{M}\epsilon \times \Phi(\epsilon)\,. 
\end{align*}
Observe that (i) $\Phi(\,\cdot\,)$ is a continuous function (by Berge's theorem), and (ii) $\Phi(0) = 0$ since \[\phi(\underline{b}, t^*(\underline{b})) = 0\,.\] Therefore, there exists $\overline{\epsilon}' > 0$ such that for all $\varepsilon \in (0, \overline{\epsilon}')$, we have \[\Phi(\epsilon) < \frac{M K}{Z}\,.\]

Now, pick any $\epsilon \in \big(0,\, \min\{\frac{1}{Z}\overline{\epsilon},\, \overline{\epsilon}'\}\big)$. We must have
\[\text{Loss}(\epsilon) \leq \frac{Z}{M}\epsilon \Phi(\epsilon) < \epsilon K = \text{Gain}(\epsilon)\,.\]
So $(a, p)$ is strictly suboptimal. A contradiction. \Cref{fig:perturb} (in \Cref{subsec:sketch-char}) illustrates this argument with a two-item example. Note that this proof holds even if we only impose monotone incremental values assumption and locally quasi-concave incremental profits assumption on the pairs of nested bundles where both bundles are in the menu $B$.

\end{document}